\documentclass[11pt]{article}
\usepackage[english,activeacute]{babel}
\usepackage[utf8]{inputenc}
\usepackage[Lenny]{fncychap}
\usepackage[hidelinks]{hyperref}
\usepackage{float}
\usepackage[ruled,vlined,linesnumbered]{algorithm2e}

\usepackage{bbm}
\usepackage{multirow}
\usepackage{amssymb,amsmath,amsthm} 
\usepackage{graphics,graphicx} 

\DeclareGraphicsExtensions{.pdf,.png,.jpg,.eps}
\usepackage{lscape}
\usepackage{amsfonts,mathrsfs}
\usepackage{multicol}
\usepackage{appendix}
\usepackage[usenames,dvipsnames]{color}
\definecolor{armygreen}{rgb}{0.29, 0.33, 0.13}
\usepackage{makeidx} 
\usepackage{appendix} 

\usepackage{txfonts}
\usepackage{subfigure}
\usepackage{float}
\usepackage{bm}
\usepackage{tikz}
\usepackage{xcolor}
\usepackage{cite}
\usepackage{bbm}
\usepackage[dvips,ps2pdf]{epsfig}
\numberwithin{figure}{section}
\numberwithin{table}{section}
\numberwithin{equation}{section}

\numberwithin{equation}{section}

\newtheorem{theorem}{Theorem}[section]

\newtheorem{lemma}[theorem]{Lema}
\newtheorem{remark}[theorem]{Remark}

\theoremstyle{remark}
\numberwithin{equation}{section}

\makeindex 
\begin{document}
\title{Analysis and optimal control of a malaria mathematical model under resistance and population movement}
\author{
\, Cristhian Montoya\footnotemark[1],\,\,\,and\, Jhoana P. Romero--Leiton \footnotemark[2]
} 

\footnotetext[1]{Institute for Mathematical and Computational Engineering. Pontificia Universidad Cat\'olica de Chile. Santiago, Chile.
{\tt cdmontoya85@gmail.com},\,\,\, \url{http://cmontoya.mat.utfsm.cl/}}

\footnotetext[2]{ Universidad de Investigaci\'on y Tecnolog\'ia Experimental Yachay Tech,
             Urcuqu\'{i}, Ecuador, 
             {\tt jpatirom3@gmail.com}}
\date{}

\maketitle
\abstract{In this work, two mathematical models for malaria under resistance are presented. More precisely,
the first model shows the interaction between humans and mosquitoes inside a patch under infection of malaria
when the human population is resistant to antimalarial drug and mosquitoes population is resistant to insecticides.
For the second model, human--mosquitoes population movements in two patches is analyzed under the same malaria
transmission dynamic established in one patch.
For a single patch, existence and stability conditions for the equilibrium solutions in terms
of the local basic reproductive number are developed. These results reveal the existence of a forward bifurcation and the
global stability of disease--free equilibrium. In the case of two patches, a theoretical and numerical
framework on sensitivity analysis of parameters is presented. After that, the use of antimalarial drugs and insecticides are incorporated as control strategies and  an optimal
control problem is formulated.  Numerical experiments are carried out in both models to show the
feasibility of our theoretical results.
}\\

\section{Introduction}\label{Section_introduction}
Malaria is a hematoprotozoan parasitic infection transmitted by certain species of anopheline mosquitoes.
Four species of plasmodium commonly infect to humans, but one, \textit{Plasmodium falciparum} is  the most
lethal in humans, causing many deaths per year. Malaria also provides an unbalance that impairs the economic and social development of certain zones of the planet \cite{guinovart2006malaria}.
In reviewing  history,  control programs have been focused in two directions: control of the anopheles mosquito
through removal of breeding sites, use of insecticides, prevention of contact with humans
(by using of screens and bed nets), and use of antimalarial drug (or effective case management) \cite{malaria2002community}. Unfortunately, the implementation of this control mechanisms has not been entirely effective.  Amongst the reasons we can mention: a) resistance of the malaria parasites to antimalarial drugs such as  chloroquine and sulfadoxine--pyrimethamine.  In this case, and from a mathematical point of view,  Aneke in \cite{aneke2002mathematical} describes the
phenomenon of antimalarial drug resistance  in a hyperendemic region by a model of ordinary
differential equations (ODEs). Esteva et al. in \cite{esteva2009qualitative} present a deterministic model for monitoring the impact of antimalarial drug resistance on the transmission dynamics of malaria in a human population. Tchuenche et al. in \cite{tchuenche2011mathematical} formulate and analyze a mathematical model for malaria with treatment and  three levels of resistance in humans incorporing both, sensitive and resistant strains of the parasites. Agusto in  \cite{agusto2014malaria}
formulates and analyzes a deterministic system of ODES for malaria transmission incorporating human movement as well as the development of antimalarial drug resistance  in a multipatch--type system. Other works to underline in this topic are  \cite{koella2003epidemiological,bacaer2005reaction,okosun2011modelling}. b) The use of pyrethroid insecticides (a man-made pesticides similar to the natural pesticide pyrethrum) in malaria vector control. Here we can find the work of Luz et al. in \cite{luz2009impact} in which a model of the seasonal population dynamics of \textit{Aedes aegypti}, both to assess the effectiveness of insecticide interventions on reducing adult mosquito abundance, and to predict evolutionary trajectories of insecticide resistance. In addition,  Aldila et al.  formulate and analyze a mathematical model for transmission of temephos resistance in \textit{Aedes aegypti} population \cite{aldila2014mathematical}, meanwhile in the works    \cite{alphey2014managing,gourley2011slowing}, the authors treat the insecticide resistance in general cases. c) The population migration problem. The movement of infected people or infected mosquitos from areas where malaria is still endemic to areas where the disease had been eradicated led to resurgence of the disease, and this situation also results in a increasing of  resistance to insecticides and antimalarial drug \cite{bloland2001drug}. With respect to migration problem, the works  have been addressed through multipatch--type models see for instance \cite{Gao-Ruan2012,Prosper2012,agusto2014malaria}. Migration problems for dengue virus and other general epidemic models have been reviewed in \cite{Hasler2016,Bichara-Abderrahman2018} and \cite{ lee2015role,Zhang2018,barrios2018assessing,Mishra-Sunita2018,Bock-Yashira2018}, respectively.
\par
As far as we know, does not exist mathematical models considering resistance to antimalarial drug and insecticides and  movement of populations simultaneously,  as factors that hinder the malaria control.  Thus, in this paper we give a first response to this situation, including numerical experiments that allow us to verify the feasibility of our theoretical results.
\par
In this paper, we propose two mathematical models for the malaria transmission dynamics and whose equations are based in  \cite{romero2018optimal}. More precisely, in the first model, we consider the interaction between humans and mosquitoes inside a patch 	when the human population is resistant to antimalarial drug and mosquitoes population is resistant to insecticides.  Existence and stability conditions for the equilibrium solutions in terms
of the local  basic reproductive number are determined. For the second model, human--mosquitoes population movements in two patches is considered  under the same conditions established in one patch and also following the ideas from  \cite{lee2015role}. Besides, by incorporating the use of antimalarial drugs and insecticides  as control strategies,  we formulate an optimal control problem for the disease.


\section{One patch model}
\label{Section_one_patch_model}
In this section, we consider a single patch with a susceptible--infected--recovered (SIR) structure for humans and a susceptible--infected (SI) structure for mosqui\-toes.
In order to present the complete model, we describe the dynamic equations that form our model as follows: let us denote as $S_{h}(t)$, $I_{h}(t)$ and $R_{h}(t)$  the number of susceptible, infected, and recovered humans at time $t$, respectively.  The total human population at time $t$ is  denoted by $N_{h}(t)=S_{h}(t)+I_{h}(t)+R_{h}(t)$.  Similarly, let us denote as $S_{v}(t)$ and $I_{v}(t)$  the number of susceptible, and infected mosquitoes at time $t$, respectively.  The total mosquito population at time $t$ is denoted by $N_{v}(t)=S_{v}(t)+I_{v}(t)$.
\newline
Moreover, from  \cite{romero2018optimal}, we define the force of infection for humans  by
$
  \beta_{h}\epsilon \frac{I_{v}}{N_{h}},
$
where $\beta_{h}$ represents the probability of a human being infected by the bite of an infected mosquito, and $\epsilon$ represents the per capita biting rate of mosquitoes. Similarly, we define the force of infection for mosquitoes  as
$
\beta_{v}\epsilon\frac{I_{h}}{N_{h}},
$
where $\beta_{v}$ represents the probability of infection of mosquito by contact with infected humans.
\par
Respect to susceptible humans population, it is increasing due to  recruitment at a constant rate of $\Lambda_{h}$ and by recovered humans from infection,  which are represented by the term $\omega R_{h}$.  Simultaneously, this population decrease due to infection by contact with infected mosquitoes through the term $\beta_{h}\epsilon \frac{I_{v}}{N_{h}}S_{h}$ and by natural death through the term $\mu_h S_h$. Thus, the ODE that represents the variation of the  susceptible humans population is
\begin{equation}\label{Sh}
   \dot {S_h}= \Lambda_{h}+\omega R_{h}-S_{h}\beta_{h}\epsilon\frac{I_{v}}{N_{h}}-\mu_{h}S_{h},
\end{equation}
where the symbol $\cdot$ corresponds to the derivative in time, i.e, $\dot{S_h}=\frac{d}{dt}S_h(t)$.
Now, respect to the infected humans population, it is treated with drug at a constant rate of $\xi_{1}\theta_{1}$, where $\xi_{1}$ is the drug  efficacy and  $\theta_{1}$ is the recovery rate due to the drug.  Besides, the number of  infected individuals  resistant to the drug (by selective pressure) is  $\xi_{1}\theta_{1}q_{1}I_{h}$, where $q_{1}\in [0,1]$ represents the resistance acquisition ratio to the drug.  Thus the term  $\xi_{1}\theta_{1}(1-q_1)I_{h}$ represents the proportion of  sensitive individuals to the drug. Additionally, a proportion of infected individuals recover spontaneously at a rate of  $\delta$ (by action of the immune system), others die from infection at a rate of  $\rho$ and others from natural death at a rate of $\mu_h$.  Thus,  the equation for the variation of the  infected humans population  is given by
\begin{equation}\label{Ih}
   \dot I_h= S_{h}\beta_{h}\epsilon\frac{I_{v}}{N_{h}}-\xi_{1}\theta_{1}(1-q_{1})I_{h}-(\delta+\rho+\mu_{h})I_{h}.
\end{equation}
Finally, in our model the recovered humans population increase by the action of the drug and by spontaneous recovery, and decrease as consequence of natural death and loss of immunity. Thus, the variation of the recovered humans population in time is described by
\begin{equation}\label{Rh}
   \dot{R}_h= \xi_1\theta_{1}(1-q_1)I_h+\delta I_{h}-(\omega+\mu_{h})R_{h}.
\end{equation}
On the other hand, the description for the SI model is the following: the susceptible mosquitoes population is recruited at a constant rate of  $\Lambda_{v}$.  It  is diminished by infection due to contact with infected humans, which is described through the term $\beta_{v}\epsilon\frac{I_{h}}{N_{h}}S_{v}$. Simultaneously, it is reduced due to natural death with a rate $\mu_{v}$ and by action of insecticides at a rate of $\xi_{2}\theta_{2}$, where $\xi_{2}$ represents the efficacy of insecticide and  $\theta_{2}$ is the death of mosquitoes due to insecticides. The number of mosquitos  resistant to the insecticides is $\xi_{2}\theta_{2}q_{2}$, with $q_{2}\in [0,1]$ represents the resistance acquisition  ratio to the insecticides. Thus, the expression $\xi_{2}\theta_{2}(1-q_{2})$ represents the proportion of sensitive mosquitos to the insecticides. Then, the system describing the variation of the mosquitoes population in time is
\begin{equation}\label{Nv}
    \left\{\begin{array}{ll}
              &\dot {S_v}= \Lambda_{v}-S_{v}\beta_{v}\epsilon\frac{I_{h}}{N_{h}}-\xi_{2}\theta_{2}(1-q_{2})S_v-\mu_{v}S_{v} \\ \\
             &\dot{ I_v}= S_{v}\beta_{v}\epsilon\frac{I_{h}}{N_{h}}-\xi_{2}\theta_{2}(1-q_{2})I_v-\mu_{v}I_{v}.
           \end{array}
     \right.
\end{equation}
\noindent
In summary, from (\ref{Sh})-(\ref{Nv}), our model for malaria under resistance in one patch is given by

\begin{equation}\label{modeloone}
  \left\{\begin{array}{ll}
  &\dot{S_{h}} = \Lambda_{h}+\omega R_{h}-S_{h}\beta_{h}\epsilon\frac{I_{v}}{N_{h}}-\mu_{h}S_{h} \\ \\
  &\dot{I_{h}} = S_{h}\beta_{h}\epsilon\frac{I_{v}}{N_{h}}-\xi_{1}\theta_{1}(1-q_{1})I_{h}-(\delta+\rho+\mu_{h})I_{h} \\ \\
  &\dot{R_{h}} = \xi_{1}\theta_{1}(1-q_1)I_h+\delta I_{h}-(\omega+\mu_{h})R_{h} \\ \\
  &\dot{S_{v}} = \Lambda_{v}-S_{v}\beta_{v}\epsilon\frac{I_{h}}{N_{h}}-\xi_{2}\theta_{2}(1-q_{2})S_v-\mu_{v}S_{v}  \\ \\
  &\dot{I_{v}} = S_{v}\beta_{v}\epsilon\frac{I_{h}}{N_{h}}-\xi_{2}\theta_{2}(1-q_{2})I_v-\mu_{v}I_{v}\\ \\
  &(\mathbf{N}_h(0),\mathbf{ N}_v(0))=(S_h(0), I_h(0). R_h(0), S_v(0), I_v(0)),
\end{array}
     \right.
\end{equation}
where $(\mathbf{N}_h(0),\mathbf{ N}_v(0))$ denotes a initial condition and $\mathbf{N}_h$ and $\mathbf{N}_v$ are vectors formed by $S_h$, $I_h$, $R_h$ and $S_v$, $I_v$, respectively.

\begin{remark}
The novelty in this work involves the parameters $\xi_i$, $\theta_i$ and $q_i$ with $i=1,2$.  Their interpretation and values are given in Tables \ref{tabla_droga} and \ref{tabla_insecticidas} from Section \ref{seccion_experimentos_one}. A complete description and interpretation of  the others parameters involved in the model (\ref{modeloone}) can be found in \cite{romero2018optimal}.
\end{remark}

Now,  a set of biological interest for the solutions of the system (\ref{modeloone}) is defined as follows
\begin{equation}\label{Omega}
    \Omega=\left\{(\mathbf{N}_h, \mathbf{N}_v)\in \mathbb{R}_5^+: \; N_h\leq \frac{\Lambda_h}{\mu_h}, \; N_v\leq \frac{\Lambda_v}{\mu_v}\right\}.
\end{equation}
The following lemma establishes the invariance property for $\Omega$.

\begin{lemma}\label{teoinvarianzaone}
For $(\mathbf{N}_h(0), \mathbf{N}_v(0))$ a non--negative initial condition,  the system (\ref{modeloone}) has a unique solution and all state variables remain non--negative for all time $t\geq 0$.  Moreover, the set defined on (\ref{Omega})
is positively invariant with respect the system (\ref{modeloone}).
\end{lemma}

\begin{proof}
Since the vector field defined on the right side of (\ref{modeloone}) is continuously differentiable, the existence and uniqueness of the solutions is fullfied.  On the other hand,
\begin{eqnarray*}
  \dot{N_{h}} &=& \Lambda_{h}-\mu_{h}N_{h}-\rho I_{h} \leq
  \Lambda_{h}-\mu_{h}N_{h}.
\end{eqnarray*}
\noindent Thus
\begin{equation*}
   \dot{N_{h}}+\mu_{h}N_{h}\leq \Lambda_{h}.
\end{equation*}
Multiplying both sides of the above inequality by the integrating factor $e^{\mu_h\tau}$ and integrating from $0$ to $t$, we obtain that
\begin{equation*}
    N_{h}(t)\leq N_{h}(0)e^{-\mu_h t}+\frac{\Lambda_{h}}{\mu_h}(1-e^{-\mu_h t}),
\end{equation*}
from where
\begin{equation*}
    \lim_{t\rightarrow \infty}N_{h(t)}\leq \frac{\Lambda_{h}}{\mu_h}.
\end{equation*}
\noindent
Similar calculation shows that $N_{v}(t)\rightarrow\frac{\Lambda_v}{\mu_v}$ as $t\rightarrow \infty$.  Thus, the region $\Omega$ is positively invariant. This complete the proof.
\end{proof}

\subsection{Qualitative analysis }
\label{Seccion_analisis_cualitativo}
In this subsection, we first compute the local basic reproductive number associated to the system (\ref{modeloone}) .  Afterward, conditions for existence and stability of the equilibrium solutions are developed.

\subsubsection{Local basic reproductive number}
\label{section_reproductivo_one}
It is well known that a  disease--free equilibrium (DFE) is a steady state solution of a system  where there is no disease, in our case,  $S_{h}=S_{h}^*>0$, $S_{v}=S_{v}^*>0$, and all others variables $I_{h}$, $I_{v}$, $R_{h}$ are zero. It will be denoted  by $\mathbf{E}_{0_{one}}=\left(\bar N_{h},0,0,\bar N_{v},0 \right)$, where
\begin{equation}\label{barN}
    \bar N_{h}=\frac{\Lambda_{h}}{\mu_{h}}, \;  \bar N_{v}=\frac{\Lambda_{v}}{\xi_2\theta_2(1-q_2)+\mu_{v}} \; \text{and} \; \mathbf{E}_{0_{one}}\in \Omega.
\end{equation}
Since the basic reproductive number, commonly denoted by $\mathcal{R}_0$ (but in this case denoted by $\mathcal{R}_{0_{one}}$) is the average number of secondary infective generated by a single infective during the curse of the infection in a whole susceptible population, it is a threshold for determining  when an outbreak can occur, or when a disease remains endemic.
Using the next generation operator method \cite{van2002reproduction} on the system (\ref{modeloone}), the Jacobian matrices $\mathbf{F}_{one}$ and $\mathbf{V}_{one}$ evaluated in the DFE are given by
\begin{equation*}
\mathbf{F}_{one}=\left(
    \begin{array}{cccc}
      0 & \beta_{h}\epsilon  \\
     \beta_{v}\epsilon\frac{\bar N_{v}}{\bar N_{h}} & 0  \\
    \end{array}
  \right)
\end{equation*}
and
\begin{equation*}
\mathbf{V}_{one}=\left(
            \begin{array}{cccc}
              \xi_1\theta_1(1-q_1)+\delta+\rho+\mu_h& 0  \\
              0 & \xi_2\theta_2(1-q_2)+\mu_v  \\
            \end{array}
          \right).
\end{equation*}
Thus, the next generator operator of model (\ref{modeloone}) is given by

\begin{equation*}\label{FV-11}
  \mathbf{F}_{one}\mathbf{V}_{one}^{-1}= \left(
    \begin{array}{cccc}
      0 & \frac{\beta_{h}\epsilon}{\xi_2\theta_2(1-q_2)+\mu_v}  \\

      \frac{\beta_{v}\epsilon \bar N_{v}}{\bar N_{h}\left(\xi_1\theta_1(1-q_1)+\delta+\rho+\mu_h\right)} & 0  \\
    \end{array}
  \right).
\end{equation*}
\noindent
It follows that  the local  basic reproduction number of the system (\ref{modeloone}), denoted by $\mathcal{R}_{0_{one}}$ is

\begin{equation}\label{R0one}
        \mathcal{R}_{0_{one}} = \left(\frac{\beta_{h}\beta_{v}\epsilon^2}{\left(\xi_1\theta_1(1-q_1)+\delta+\rho+\mu_h  \right)\left(\xi_2\theta_2(1-q_2)+\mu_v    \right)}\frac{\bar N_{v}}{\bar N_{h}}\right)^{1/2}.
\end{equation}


\subsubsection{Existence of endemic equilibria}
\label{seccion_equilibrios_endemicos}

In this subsection, conditions for existence of endemic equilibria of the model (\ref{modeloone}) are studied. First of all, the existence of the DFE, denoted by $\mathbf{E}_{0_{one}}$, is guaranteed as consequence of the previous subsection. Now, in order to analyze the endemic equilibria of the model (\ref{modeloone}) we  consider the solutions to the  algebraic equation system
\begin{equation}\label{ecuaciondeequilibrio}
   \left\{  \begin{array}{rll}
    \Lambda_{h}+\omega R_{h}-S_{h}\beta_{h}\epsilon\frac{I_{v}}{N_{h}}-\mu_{h}S_{h}&=&0 \\ \\
   S_{h}\beta_{h}\epsilon\frac{I_{v}}{N_{h}}-\xi_{1}\theta_{1}(1-q_{1})I_{h}-(\delta+\rho+\mu_{h})I_{h} &=&0\\ \\
   \xi_{1}\theta_{1}(1-q_1)I_h+\delta I_{h}-(\omega+\mu_{h})R_{h} &=&0 \\ \\
  \Lambda_{v}-S_{v}\beta_{v}\epsilon\frac{I_{h}}{N_{h}}-\xi_{2}\theta_{2}(1-q_{2})S_v-\mu_{v}S_{v} &=&0  \\ \\
   S_{v}\beta_{v}\epsilon\frac{I_{h}}{N_{h}}-\xi_{2}\theta_{2}(1-q_{2})I_v-\mu_{v}I_{v}&=&0.
    \end{array} \right.
\end{equation}
Let us define
\begin{equation}\label{alpha}
    \alpha=\frac{\omega}{\omega+\mu_h}\frac{
\xi_1\theta_1(1-q_1)+\delta}{
\xi_1\theta_1(1-q_1)+\delta+\rho+\mu_h}\; \text{where} \; \alpha<1.
\end{equation}
Thus, after some algebraic manipulations of the system (\ref{ecuaciondeequilibrio}), we obtain the following expressions for $S_h$, $R_h$, $S_v$ and $I_v$ in terms of $I_h$
\begin{equation}\label{solucionesdeequilirio}
    \begin{array}{ll}
      & S_h=\frac{\Lambda_h}{\mu_h}-\frac{
\xi_1\theta_1(1-q_1)+\delta+\rho+\mu_h}{\mu_h}(1-\alpha)I_h \\ \\
      & R_h=\frac{\xi_1\theta_1(1-q_1)+\delta}{\omega+\mu_h}I_h \\ \\
      &S_v=\bar N_v\left(1-\frac{\beta_v\epsilon \bar N_v}{(
\xi_2\theta_2(1-q_2)+\mu_v)\bar N_h+\left[\beta_v\epsilon+\frac{\rho}{\mu_h}(
\xi_2\theta_2(1-q_2)+\mu_v)  \right]}I_h \right) \\ \\
      &I_v=\frac{\beta_v\epsilon \bar N_v}{(
\xi_2\theta_2(1-q_2)+\mu_v)\bar N_h+\left[\beta_v\epsilon+\frac{\rho}{\mu_h}(
\xi_2\theta_2(1-q_2)+\mu_v)  \right]}I_h,
    \end{array}
\end{equation}
and the  following cuadratic equation for $I_h$
\begin{equation}\label{cuadratica}
    aI_h^2+bI_h+c=0, \; \text{where}
\end{equation}

\begin{equation}\label{coeficientescuad}
\begin{array}{ll}
          a=&  \frac{\rho}{\mu_h}\left[\beta_v\epsilon+\frac{\rho}{\mu_h}\left(\xi_2\theta_2(1-q_2)+\mu_v  \right)\right] \\ \\
  b =& a\bar N_h+\frac{(1-\alpha)\bar N_h}{\mu_h}\left( \xi_2\theta_2(1-q_2)+\mu_v \right)\left(
\xi_1\theta_1(1-q_1)+\delta+\rho+\mu_h \right)\times \\ \\ & \left[\frac{\rho}{(1-\alpha)(
\xi_1\theta_1(1-q_1)+\delta+\rho+\mu_h)}+\mathcal{R}_{0_{one}}^2 \right] \\ \\
  c =& \left(\xi_2\theta_2(1-q_2)+\mu_v  \right)\bar N_h^2\left[1- \mathcal{R}_{0_{one}}^2 \right].
       \end{array}
\end{equation}
\noindent
From (\ref{coeficientescuad}), we have e that the coefficients $a$ and $b$ are non--negatives, while $c\geq 0$ if $\mathcal{R}_{0_{one}}^2\leq 1$, otherwise $c<0$.  Thus, the polynomial $P(I_h)=aI_h^2+bI_h+c$ has only one sign change and by the Descartes' rule of sign \cite{anderson1998descartes} it has one or zero positive roots.
\newline
This result is summarized in the following theorem.

\begin{theorem}\label{teoendemicas}
For the model (\ref{modeloone}) always exists the DFE contained in $\Omega$.  Additionally,
\noindent\begin{enumerate}
\item If $\mathcal{R}_{0_{one}}\leq 1$, there are not endemic equilibria.
\item If $\mathcal{R}_{0_{one}}>1$ there exist one endemic equilibrium .
 \end{enumerate}
\end{theorem}
\subsubsection{Stability analysis}
\label{seccion_analisis_estabilidad}

In this subsection, we proof the stability of the equilibrium solutions of the system (\ref{modeloone}) given on Theorem \ref{teoendemicas}.  First, using the  linearization of the system (\ref{modeloone}) at the DFE, we proof it local stability, which is determined by the sign of the real part of the eigenvalues of the Jacobian matrix denoted by $J(\mathbf{E}_{0_{one}})$, which is given by
{\footnotesize
\begin{equation}\label{jacobiangeneral}
  J(\mathbf{E}_{0_{one}})=  \left(
        \begin{array}{ccccc}
        -\mu_h & 0 & \omega & 0 & -\beta_h\epsilon \\ \\
          0      & -\left[\xi_1\theta_1(1-q_1)+\delta+\rho+\mu_h  \right] & 0 & 0 & \beta_h\epsilon \\ \\
          0 & \xi_1\theta_1(1-q_1)+\delta & -(\omega+\mu_h) & 0 & 0 \\ \\
          0 & -\beta_v\epsilon\frac{\bar N_v}{\bar N_h} & 0 & -\left[\xi_2\theta_2(1-q_2)+\mu_v \right] & 0 \\ \\
          0 & \beta_v\epsilon\frac{\bar N_v}{\bar N_h} & 0 & 0 & -\left[\xi_2\theta_2(1-q_2)+\mu_v \right] \\ \\
        \end{array}
      \right),
\end{equation}}
\noindent

Three eigenvalues of $J(\mathbf{E}_{0_{one}})$ are $\eta_1=-\mu_h$, $\eta_2=-(\omega+\mu_h)$ and $\eta_3=-\left[\xi_2\theta_2(1-q_2)+\mu_v \right]$, while the others eigenvalues are given by the roots of the following quadratic equation
\begin{equation}\label{cuadratica2}
    a_0\eta^2+a_1\eta+a_2=0, \; \text{where}
\end{equation}
\begin{eqnarray*}
  a_0 &=& 1 \\
  a_1 &=& \xi_2\theta_2(1-q_2)+\mu_v+\xi_1\theta_1(1-q_1)+\delta+\rho+\mu_h\\
  a_2 &=& \left(\xi_2\theta_2(1-q_2)+\mu_v \right)\left(\xi_1\theta_1(1-q_1)+\delta+\rho+\mu_h \right)(1-\mathcal{R}_{0_{one}}^2).
\end{eqnarray*}
From above, the coefficients $a_0$ and $a_1$ are positives, while the sign of the coefficient $a_2$ depends of $\mathcal{R}_{0_{one}}$.  From the Routh--Hurwitz criterion \cite{dejesus1987routh} we can guarantee that the quadratic equation (\ref{cuadratica2}) has roots with  negative real part if and only if its coefficients are positives and the following determinants called minors of Hurwitz are positives
\begin{eqnarray*}
  \Delta_1 &=& a_1 \\
  \Delta_2 &=& \left|
                 \begin{array}{cc}
                   a_1 & 1 \\
                   0 & a_2 \\
                 \end{array}
               \right|=a_1a_2.
\end{eqnarray*}
We verify that $\Delta_1>0$ and $\Delta_2>0$ if and only if  $\mathcal{R}_{0_{one}}\leq 1$.  In consequence, when $\mathcal{R}_{0_{one}}\leq 1$ the DFE is a locally asymptotically stable (LAS) equilibrium point of the system (\ref{modeloone}).
\par

Now, we are going to proof the stability of the endemic equilibrium of the system (\ref{modeloone}).  For this end,  we use results based on the center manifold theory described in \cite{castillo2004dynamical} to show that the system (\ref{modeloone})  exhibits a forward bifurcation when $\mathcal{R}_{0_{one}}=1$ or equivalently when
\begin{equation}\label{eq}
    \beta_h\epsilon\doteq\beta^*=\frac{\bar N_h \left( \xi_2\theta_2(1-q_2)+\mu_v \right)\left(\xi_1\theta_1(1-q_1)+\delta+\rho+\mu_h  \right)}{\beta_v\epsilon\bar N_v}.
\end{equation}
\noindent
The eigenvalues of the Jacobian matrix given on  (\ref{jacobiangeneral}) evaluated in $(\mathbf{E}_{0_{one}},\beta^*)$
are $0$ and  $-\mu_h$, \, $-(\omega+\mu_h)$, \, $-\left[\xi_2\theta_2(1-q_2)+\mu_v\right]$ \, and $-\left[\xi_2\theta_2(1-q_2)+\mu_v+\xi_1\theta_1(1-q_1)+\delta+\rho+\mu_h  \right]$,\, where the last four have negative real part.
In consequence, in $\beta^*$,  the DFE is a non--hyperbolic equilibrium.
\noindent
Let  $\mathbf{W}=(w_1, w_2, w_3, w_4, w_5)^T$ a right eigenvector associated  to the zero eigenvalue, which satisfies $J(\mathbf{E}_{0_{one}},\beta^*)\mathbf{W}=0\mathbf{W}=\textbf{0}$ or equivalently
\begin{equation*}
    \left\{\begin{array}{rl}
               -\mu_hw_1+\omega w_3-\beta^*w_5&=0 \\\\
         -\left(\xi_1\theta_1(1-q_1)+\delta+\rho+\mu_h\right)w_2+\beta^*w_5&=0 \\ \\
         -(\omega+\mu_h)w_3+\frac{\left(
\xi_2\theta_2(1-q_2)+\mu_v  \right)\left(\xi_1\theta_1(1-q_1)+\delta  \right)\bar N_h}{\beta_v\epsilon\bar N_v}w_5&=0 \\\\
 -\left( \xi_2\theta_2(1-q_2)+\mu_v \right)w_4-\left( \xi_2\theta_2(1-q_2)+\mu_v \right)w_5&=0. \\
           \end{array}
    \right.
\end{equation*}
The vectorial form for the solutions of above linear system is given by

\begin{equation}\label{W}
    \mathbf{W}=\left[ \frac{\beta_h\epsilon}{\mathcal{R}_{0_{one}}^2}(1-\alpha), \frac{\left(\xi_2\theta_2(1-q_2)+\mu_v  \right)\bar N_h}{\beta_v\epsilon\bar N_v}, \frac{\left( \xi_2\theta_2(1-q_2)+\mu_v \right)\left( \xi_1\theta_1(1-q_1)+\delta \right)\bar N_h}{\beta_v\epsilon(\omega+\mu_h)\bar N_v},-1,1\right]^Tw_5,
\end{equation}
\noindent
where the parameter $\alpha$ is defined on (\ref{alpha}). Similarly, a left eigenvector $\mathbf{V}=(v_1, v_2, v_3, v_4, v_5)$  of the matrix $J(\mathbf{E}_{0_{one}},\beta^*)$ associated to the zero eigenvalue satisfies that $\mathbf{V}J(\mathbf{E}_{0_{one}},\beta^*)=0\mathbf{V}=\textbf{0}$ or equivalently $v_1=v_3=v_4=0$ and
\begin{eqnarray*}
  \left(\xi_1\theta_1(1-q_1)+\delta+\rho+\mu_h  \right)v_2+\beta_v\epsilon\frac{\bar N_v}{\bar N_h}v_5 &=& 0,
\end{eqnarray*}
from where
\begin{equation}\label{V}
    \mathbf{V}=\left[0,\frac{\beta_v\epsilon \bar N_v}{\bar N_h \left(\xi_1\theta_1(1-q_1)+\delta+\rho+\mu_h  \right)}, 0,0,1    \right]v_5.
\end{equation}
\noindent
The values for $w_5$ and $v_5$ such that $\mathbf{W} \cdot\mathbf{V} =1$, are
\begin{equation}\label{w2v5}
    w_5=1 \; \text{and} \; v_5= \frac{\xi_1\theta_1(1-q_1)+\delta+\rho+\mu_h}{\xi_1\theta_1(1-q_1)+\delta+\rho+\mu_h+\xi_2\theta_2(1-q_2)+\mu_v}.
\end{equation}
\noindent
Thus, the coefficients $\tilde a$ and $\tilde b$ given on Theorem 4.1 from \cite{castillo2004dynamical}
\begin{equation}\label{tildesgenerales}
   \begin{array}{ll}
     \tilde a =& \frac{1}{2}\sum_{k,i,j=1}^5 v_kw_iw_j \frac{\partial^2f_k}{\partial x_i\partial x_j}(\mathbf{E}_{0_{one}}, \beta^*) \\ \\
     \tilde b =& \sum_{k,i=1}^5 v_kw_i \frac{\partial^2f_k}{\partial x_i\partial \beta^*}(\mathbf{E}_{0_{one}}, \beta^*),
   \end{array}
\end{equation}
\noindent
can be explicitly computed as follows. Let us denote as $f_i$, $i=1,...,5$ to the scalar functions of the right hand of the system (\ref{modeloone}),  and  $x_1=S_h$, $x_2=I_h$, $x_3=R_h$, $x_4=S_v$, $x_5=I_v$. The coefficients $w_p$ and $v_p$ with $p=1,2,...5$ of (\ref{tildesgenerales}), represent to the components of the eigenvectors $\mathbf{W}$ and $\mathbf{V}$ defined on (\ref{W}) and (\ref{V}), respectively.  After some calculations we have that the second order partial derivatives evaluated in $(\mathbf{E}_{0_{one}}, \beta^*)$ are given by
\begin{equation*}
    \begin{array}{lll}
      \frac{\partial^2f_1}{\partial x_5\partial x_2}= \frac{\partial^2f_1}{\partial x_5\partial x_3} = \frac{\beta_h\epsilon}{\bar N_h} & \frac{\partial^2f_5}{\partial x_2\partial x_1}= \frac{\partial^2f_5}{\partial x_2\partial x_3} = \frac{-\beta_v\epsilon\bar N_v}{\bar N_h^2} & \frac{\partial^2f_2}{\partial x_5\partial x_2}= \frac{\partial^2f_2}{\partial x_5\partial x_3}= -\frac{\beta_h\epsilon}{\bar N_h} \\ \\
      \frac{\partial^2f_4}{\partial x_2\partial x_1} = \frac{\beta_v\epsilon\bar N_v}{\bar N_h^2} & \frac{\partial^2f_4}{\partial x_2\partial x_4} = -\frac{\beta_v\epsilon}{\bar N_h} & \frac{\partial^2f_5}{\partial x_2\partial x_4} = \frac{\beta_v\epsilon}{\bar N_h} \\ \\
      \frac{\partial^2f_4}{\partial x_2^2} =\frac{2\beta_v\epsilon\bar N_v}{\bar N_h^2} & \frac{\partial^2f_5}{\partial x_2^2} =-\frac{2\beta_v\epsilon\bar N_v}{\bar N_h^2}. &
    \end{array}
\end{equation*}
\noindent
In the above expressions  we did not consider to the zero  and   cross partial derivatives. Additionally, the second order partial derivatives with respect to the  bifurcation parameter $\beta^*$ evaluated in  $\mathbf{E}_{0_{one}}$ are all zero except
\begin{equation*}
  \frac{\partial^2f_1}{\partial x_5\partial \beta^*}=-1 \quad \text{and} \quad
   \frac{\partial^2f_2}{\partial x_5\partial \beta^*} = 1.
\end{equation*}
Thus, the coefficients $\tilde a$ and $\tilde b$ given on (\ref{tildesgenerales})
can be expressed as

\begin{equation}\label{tildes}
   \begin{array}{ll}
     \tilde a &= 2\frac{\partial^2f_2}{\partial x_5\partial x_2}v_2(w_2+w_3)+2\frac{\partial^2f_5}{\partial x_1\partial x_2}w_2(w_1+w_3)-\frac{\partial^2f_5}{\partial x_2\partial x_4}w_2+\frac{1}{2}\frac{\partial^2f_5}{\partial x_2^2}w_2^2 \\ \\
  &= -\frac{2\beta^*}{\bar N_h}v_2(w_2+w_3)-\frac{2\beta_v\epsilon\bar N_v} {\bar N_h^2}w_2(w_1+w_3)-\frac{\beta_v\epsilon}{\bar N_h}w_2-\frac{\beta_v\epsilon\bar N_v}{\bar N_h^2}w_2^2 \\ \\
     \tilde b &= v_2w_5\frac{\partial^2f_2}{\partial x_5\partial \beta^*}=\frac{\beta_v\epsilon\bar N_v}{\bar N_h\left(\xi_1\theta_1(1-q_1)+\delta+\rho+\mu_h+\xi_2\theta_2(1-q_2)+\mu_v  \right)}.
   \end{array}
\end{equation}
\noindent
From (\ref{tildes}) we have that  $\tilde b>0$ while the sign of $\tilde a$ depends of the sign of $w_2$, $v_2(w_2+w_3)$ and $w_2(w_1+w_3)$.  From  (\ref{W}) and (\ref{V}) we verify that $w_2\geq 0$, $v_2(w_2+w_3)\geq 0$,  and
\begin{equation*}
    w_2(w_1+w_3)=w_2\left[\frac{\beta_h\epsilon}{\mathcal{R}_{0_{one}}^2}(1-\alpha)+\frac{\left( \xi_2\theta_2(1-q_2)+\mu_v \right)\left( \xi_1\theta_1(1-q_1)+\delta \right)\bar N_h}{\beta_v\epsilon(\omega+\mu_h)\bar N_v}\right]\geq 0.
\end{equation*}
\noindent
Thus, by Theorem 4.1 from \cite{castillo2004dynamical}, the endemic equilibrium is LAS when $\mathcal{R}_{0_{one}}>1$, which suggest the global stability of the DFE. The previous results are summarized in the following theorem.

\begin{theorem}\label{teorema_estabilidadlocal}
If $\mathcal{R}_{0_{one}}\leq 1$ the DFE is LAS in $\Omega$, and the endemic equilibrium is unestable. If $\mathcal{R}_{0_{one}}>1$ the DFE becomes an unstable hyperbolic equilibrium point, and the endemic equilibrium is LAS in $\Omega$.
\end{theorem}
\noindent
Figure \ref{bifurcation} shows the bifurcation diagram.

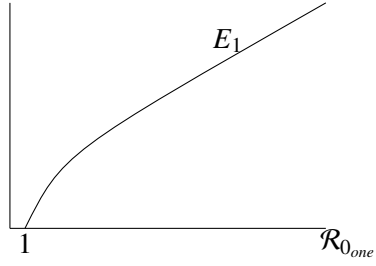
\begin{figure}
\begin{center}
\begin{tikzpicture}
\draw (-1.2,0) -- (3,0);
\draw (-1.2,0) -- (-1.2,3);
\draw (3.3,-0.2) node{$\mathcal{R}_{0_{one}}$};
\draw (-1,-0.2) node{1};
\draw (1.7,2.5) node{$E_1$};
\draw (-1,0) .. controls  (-0.5,1) .. (3,3);
\end{tikzpicture}
\caption{ {\scriptsize A forward bifurcation occurs  when $\mathcal{R}_{0_{one}}=1$.}} \label{bifurcation}
\end{center}
\end{figure}

\begin{theorem}\label{teoglobalDFE}
If $\mathcal{R}_{0_{one}}^2\leq 1$, then the DFE is globally asymptotically (GAS) stable in $\Omega$.
\end{theorem}

\begin{proof}
From Theorem \ref{teorema_estabilidadlocal}, when $\mathcal{R}_{0_{one}}^2\leq 1$,  the $DFE$ is LAS in $\Omega$.  Let $(\mathbf{N}_h(t), \mathbf{N}_v(t))$ a positive solution of the system (\ref{modeloone}), then by Lemma \ref{teoinvarianzaone} it satisfies that
\begin{equation}\label{eq01}
    S_v(t)\leq \bar N_v(t) \; \text{and} \; N_h(t)>\frac{\Lambda_h}{\rho+\mu_h}.
\end{equation}
\noindent
We will proof the existence of a Lyapunov function for the traslated system $\dot{ \mathbf{y}}=f(\mathbf{y}+\mathbf{E}_{0_{one}})-f(\mathbf{E}_{0_{one}})=F(\mathbf{y})$, where $f$ is the vectorial field defined from right hand of the system (\ref{modeloone}) and $\mathbf{y}=\mathbf{0}$ is a trivial solution of the system $\dot {\mathbf{y}}=F(\mathbf{y})$.  Let us consider the following function
\begin{equation*}
    V^*(\mathbf{N}_h, \mathbf{N}_v)=\frac{\beta_h\epsilon}{\xi_2\theta_2(1-q_2)+\mu_v}I_v+I_h,
\end{equation*}
and let
\begin{equation}\label{lyapunovfunction}
    V(\tilde {\mathbf {N}}_h, \tilde{ \mathbf {N}}_v)=V^*(S_h-\bar N_h, I_h, R_h, S_v-\bar N_v, I_v).
\end{equation}
The function  $V$ defined on (\ref{lyapunovfunction}) satisfies the following properties
\begin{itemize}
  \item [(P1)] $V(\bar N_h, 0, 0, \bar N_v, 0)=V(\mathbf{E}_{0_{one}})=V^*(\mathbf{0})=0$.
  \item [(P2)] $V>0$ $\forall(\tilde{\mathbf{N}_h},\tilde{\mathbf{N}_v} )\neq \mathbf{E}_{0_{one}}$ in $\Omega$ (V is positive definite).
  \item [(P3)] The orbital derivative of $V$  along the trajectories of (\ref{modeloone}) is negative definite.  In fact,
      \begin{eqnarray*}
        \dot V &=& \frac{\partial V^*}{\partial (S_h-\bar N_h)}f_1+\frac{\partial V^*}{\partial I_h}f_2+\frac{\partial V^*}{\partial R_h}f_3+\frac{\partial V^*}{\partial (S_v-\bar N_v)}f_4+\frac{\partial V^*}{\partial I_v} \\
         &=&S_{h}\beta_{h}\epsilon\frac{I_{v}}{N_{h}}-\xi_{1}\theta_{1}(1-q_{1})I_{h}-(\delta+\rho+\mu_{h})I_{h}  +\\
         &&\frac{\beta_h\epsilon}{\xi_2\theta_2(1-q_2)+\mu_v}\left(S_{v}\beta_{v}\epsilon\frac{I_{h}}{N_{h}}-\xi_{2}\theta_{2}(1-q_{2})I_v-\mu_{v}I_{v}\right) \\
         &=& \left[\frac{S_v\beta_h\beta_v\epsilon^2}{N_h\left( \xi_2\theta_2(1-q_2)+\mu_v \right)}-\left(\xi_1\theta_1(1-q_1)+\delta+\rho+\mu_h \right)\right]I_h\\
         &&+\left[\beta_h\epsilon\frac{S_h}{N_h}-\beta_h\epsilon  \right]I_v \\
         &\leq& \left(\xi_1\theta_1(1-q_1)+\delta+\rho+\mu_h  \right)\left[\mathcal{R}_{0_{one}}^2-1  \right]I_h \leq 0.
      \end{eqnarray*}
\end{itemize}
Thus, the DFE is globally stable in $\Omega$. To verify its  global asymptotic stability, let us consider $\triangle =\{ (\mathbf{N}_h, \mathbf{N}_v): \dot{V}=0 \}$.
Then $\triangle \subset \{ (\mathbf{N}_h, \mathbf{N}_v): I_h=0\} $.  Let $\triangle^\prime \subset \triangle$ the biggest invariant set with respect to (\ref{modeloone}) and $(\mathbf{N}_h, \mathbf{N}_v)$ a solution of (\ref{modeloone}) in $\triangle^\prime$, then  $(\mathbf{N}_h, \mathbf{N}_v)$ is defined and is bounded $\forall t\in  \mathbb{R}$ and $I_h(t)= 0$ in $\triangle^\prime$ for all $t$. Replacing this value in the system  (\ref{modeloone}) we obtain that $R_h(t)=I_v(t) = 0$ for all $t$, while from the first and fourth equation of  (\ref{modeloone}) we obtain that  $S_h(t)= \Lambda_h/\mu_h=\bar N_h$ and $S_v(t) = \Lambda_v/(\xi_2\theta_2(1-q_2)+\mu_v)=\bar N_v$.  Thus,  $\triangle^\prime=\{\mathbf{E}_{0_{one}}\}$ and from the Lasalle invariance principle \cite{wei2011controlling} $\mathbf{E}_{0_{one}}$ is GAS in $\Omega$.
\end{proof}

\subsection{Numerical experiments}
\label{seccion_experimentos_one}

In this subsection, we validate our theoretical results with numerical experiments.  For this end, we take data from rural areas of  Tumaco (Colombia) reported in \cite{romero2018optimal} and make some numerical simulations.  For the values of the parameters corresponding to insecticides, we assume that the fumigation is done with two  pyrethroids insecticides (deltamethrin and cyfluthrin) according to the recommendations of Palomino et al. in \cite{palomino2008eficacia}. Pyrethroids insecticides are  a special chemicals class of active ingredients found in many of the modern insecticides used by pest management professionals. Due to the low concentrations in which these products are applied, a constant safety of use and a decrease in the toxic impact on vector control have been achieved.  For the values of the parameters corresponding to the drug, we assume that the infected patients are treated with artemisinin--based combination therapy (ACT) according to the  recomendations of Smith in \cite{smith2018efficacy}. Artemisinin (also called qinghaosu), is an antimalarial drug derived from the sweet wormwood plant: \textit{Artemisia annua}. Fast acting artemisinin--based compounds are combined with other drugs, for example, lumefantrine, mefloquine, amodiaquine, sulfadoxine/pyrimethamine, piperaquine and chlorproguanil/dapsone. The artemisinin derivatives include dihydroartemisinin, artesunate and artemether \cite{smith2018efficacy}.   Tables \ref{tabla_droga} and \ref{tabla_insecticidas} show the values of the parameters corresponding to  the drugs and insecticides supply, respectively.

\begin{table}[H]
\begin{center}
\caption{Values of the parameters corresponding to the ACT  supply.}
\label{tabla_droga}
\begin{tabular}{llll}\hline
  Parameter    & Interpretation     &Dimension                      &  Value          \\ \hline
  $\xi_1$       & Drug efficacy   &Dimensionless                     &   0.7\\
  $\theta_1$    & Recovery rate due to the drug  &Day$^{-1}$              & 0.6\\
  $q_1$         & Resistance acquisition ratio to the drug &
Dimensionless & 0.1 \\\hline
\end{tabular}
\end{center}
\end{table}
\begin{table}[H]
\begin{center}
\caption{ Values of the parameters corresponding to insecticides  supply.}
\label{tabla_insecticidas}
\begin{tabular}{lll}\hline
  Parameter    & Interpretation       &Dimension \\                       \\ \hline
  $\xi_2$       & Insecticide efficacy   &Dimensionless    \\
  $\theta_2$    & Death rate due to the insecticides &Day $^{-1}$     \\
  $q_2$         & Resistance acquisition ratio to the insecticides &Dimensionless \\ \hline
               &Value for deltamethrin  & Value for cyfluthrin     \\ \hline
               &0.7              &0.2       \\
               & 0.3     &0.3 \\
               & 0.05 & 0.2 \\ \hline
\end{tabular}
\end{center}
\end{table}

Figure \ref{modelosincontrol2_unparche} shows the behavior of human and mosquito populations when the patients are treated with ACT and the mosquitoes are fumigated with cyfluthrin and deltamethrin, respectively. In Figure \ref{modelosincontrol2_unparche} (a) the solutions tend to an endemic equilibrium and $\mathcal{R}_{0_{one}}^2=2.15$, while in Figure \ref{modelosincontrol2_unparche} (b) the solutions tend to the DFE and   $\mathcal{R}_{0_{one}}^2=0.0012$.  In fact, given that cyfluthrin  is an insectcide with less efficacy  than deltamethrin, its application generates greater resistance hindering the disease control.

\begin{figure}[H]
\begin{center}
\subfigure[Fumigation with cyfluthrin]{\includegraphics[width=7 cm, height=7 cm]{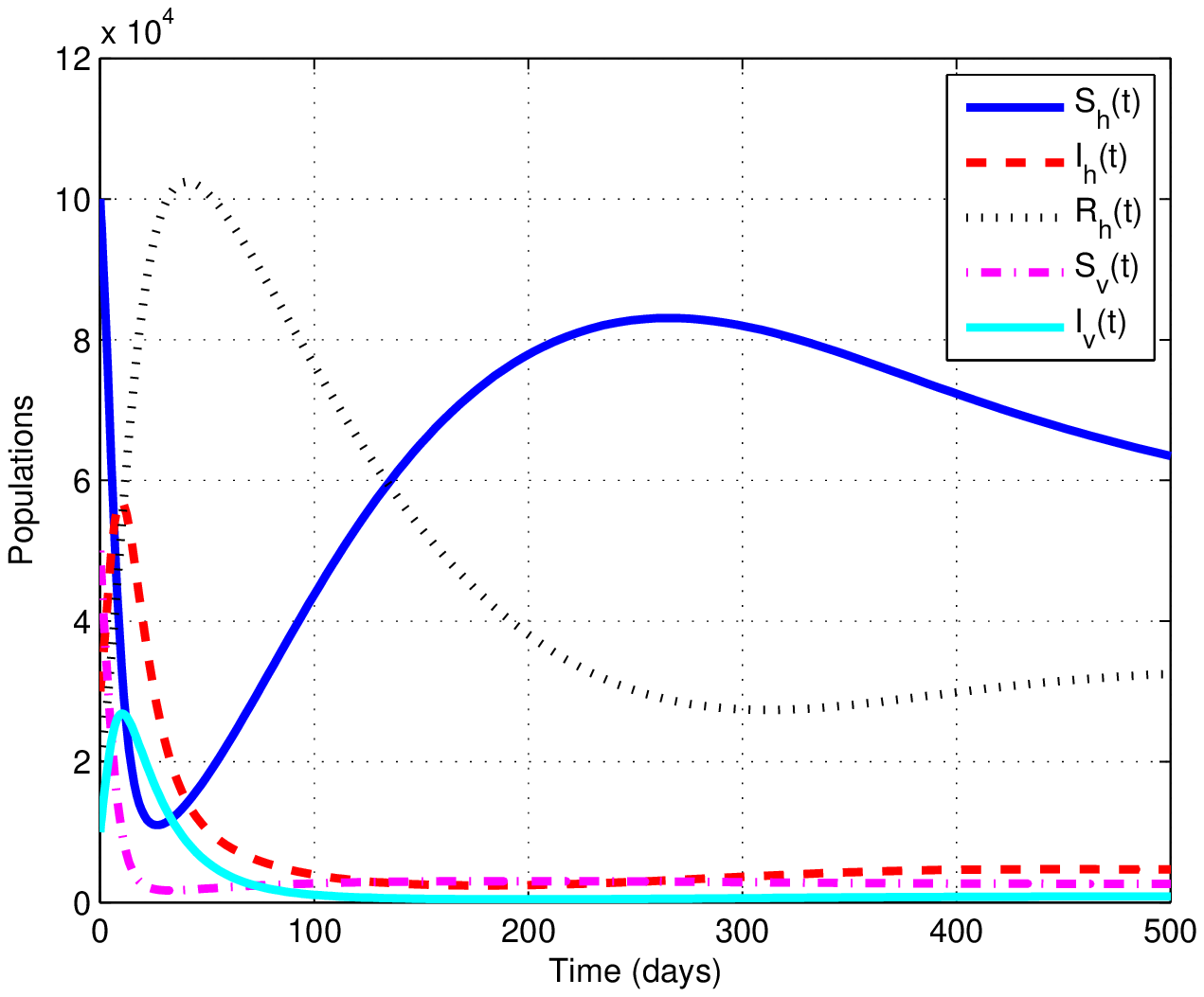}}
\subfigure[ Fumigation with deltamethrin]{\includegraphics[width=7 cm, height=7 cm]{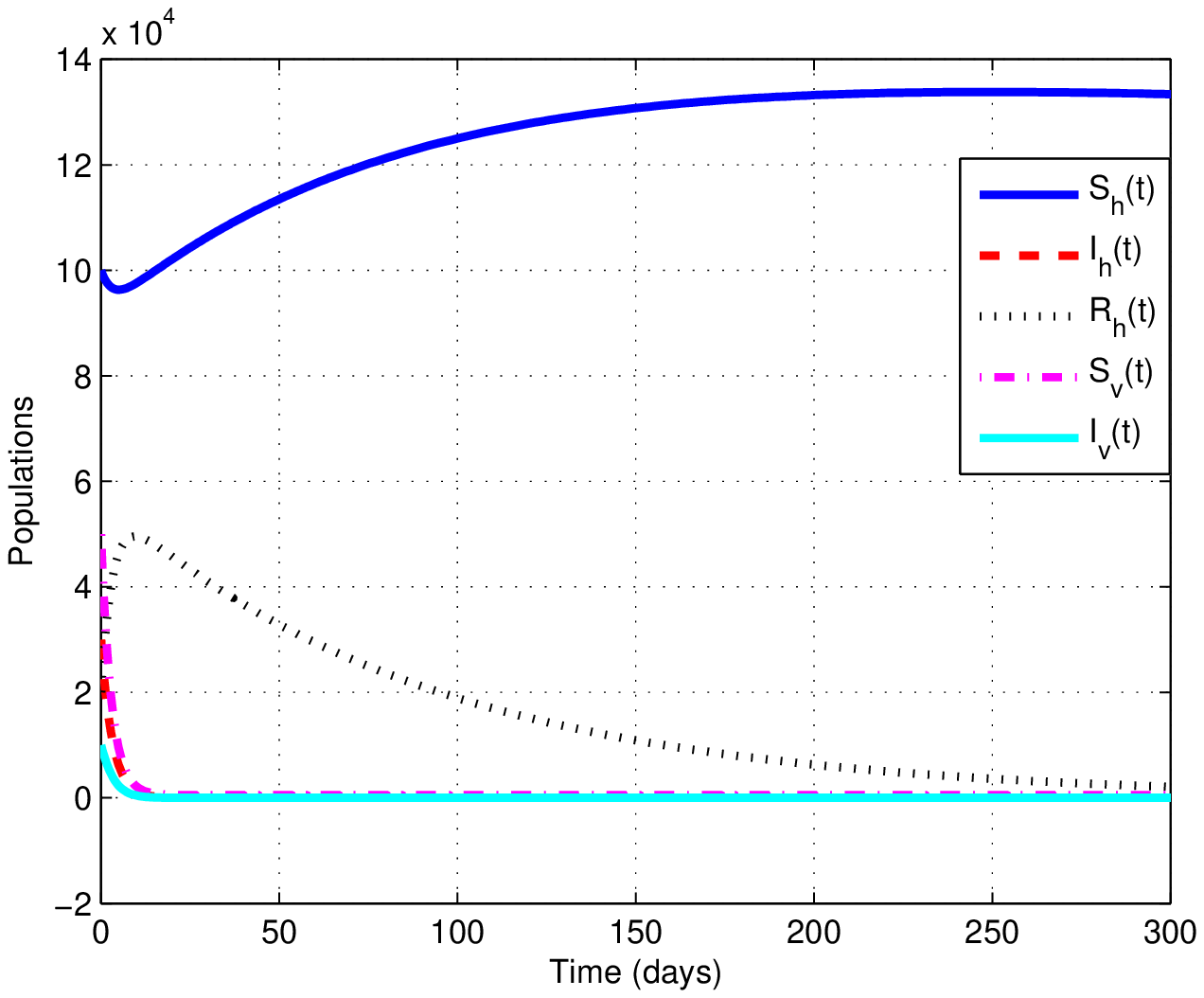}}
\end{center}
\caption{ Numerical simulations of model (\ref{modeloone}) with data from rural areas of Tumaco (Colombia) reported in \cite{romero2018optimal} and  initial condition (100000, 30000, 20000,  50000, 10000). On the left, the fumigation is done with cyfluthrin,  here $\mathcal{R}_{0_{one}}^2=2.15$ and the solutions tend to the endemic equilibrium  (63480, 4690, 32480, 2630, 840). On the right, the fumigation is done with deltamethrin, $\mathcal{R}_{0_{one}}^2=0.0012$ and  the solutions tend to the DFE.}
\label{modelosincontrol2_unparche}
\end{figure}

In Figure \ref{comparacion} we consider the effects of  resistance in the population dynamics.  In Figures \ref{comparacion} (a) and  (b) we assume that there is no resistance ($q_1=q_2=0$). Then, when the fumigation is done  with cyflutrin,  $\mathcal{R}_{0_{one}}^2=1.41$ and the solutions tend to the endemic equilibrium (8136, 227, 1541, 250, 32), which evidences a considerable reduction in the persistence of the infection, while  if the fumigation is done with deltamethrin, $\mathcal{R}_{0_{one}}^2=0.00095$ and the solutions tend to the DFE. In Figures \ref{comparacion} (c) and  (d) we assume total resistance ($q_1=q_2=1$).  Then, when the fumigation is done with cyflutrin, $\mathcal{R}_{0_{one}}^2=2724.4$ and the solutions tend to the endemic equilibrium (0, 789, 0, 0, 4699), which evidences that after the first 30 days, all individuals (humans and mosquitoes) will be infected, while if the fumigation is done with deltamethrin, $\mathcal{R}_{0_{one}}^2=264.8$ and the solutions tend to the endemic equilibrium (6, 824, 0, 195, 4617), which evidences a persistence of the infection.

\begin{figure}[H]
\begin{center}
\subfigure[ Fumigation with cyfluthrin and $q_1=q_2=0$]{\includegraphics[width=7 cm, height=7 cm]{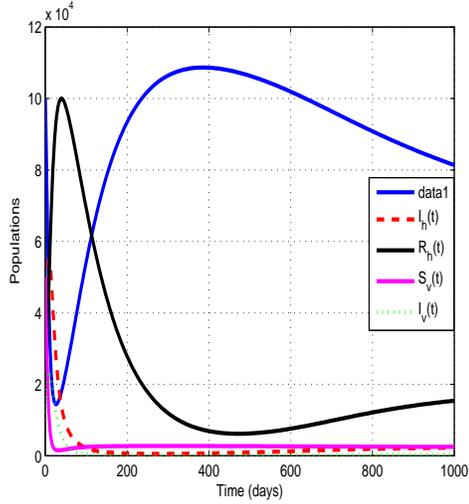}}
\subfigure[ Fumigation with deltamethrin and $q_1=q_2=0$]{\includegraphics[width=7 cm, height=7 cm]{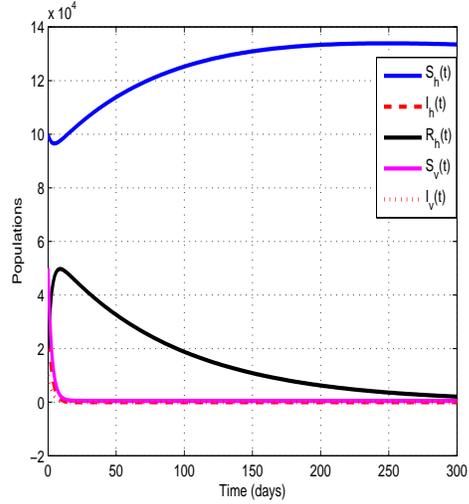}}
\subfigure[Fumigation with cyfluthrin and $q_1=q_2=1$]{\includegraphics[width=7 cm, height=7 cm]{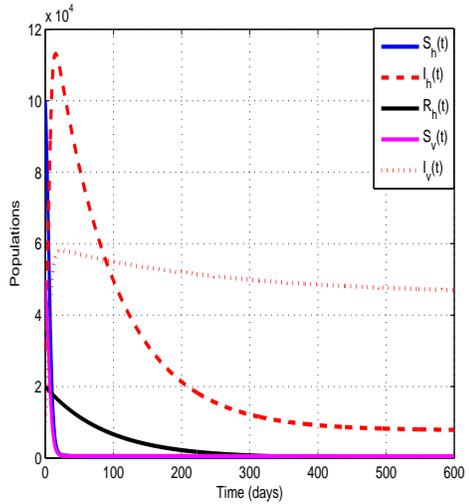}}
\subfigure[Fumigation with deltamethrin and $q_1=q_2=1$]{\includegraphics[width=7 cm, height=7cm]{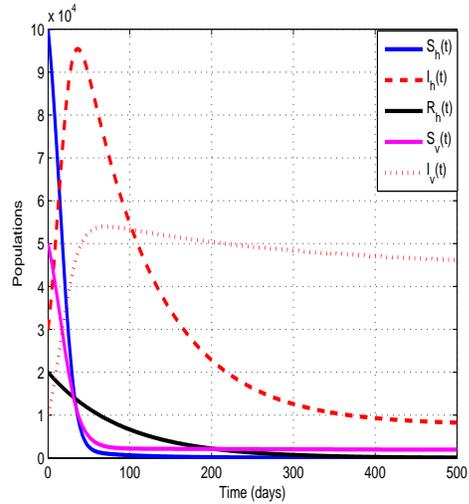}}
\end{center}
\caption{ Total resistance ($q_1=q_2=1$) and no resistance ($q_1=q_2=0$).}
\label{comparacion}
\end{figure}

\section{Two patch model}
\label{Seccion_dos_parches}
In this section we model the malaria transmission dynamics between humans and mosquitoes within a patch  and their spatial dispersal between two patches. Within a single patch, our model is defined by the equations (\ref{modeloone}), where the subscripts $1$ and $2$ refers to patch $1$ and patch $2$, respectively. The patches are coupled via the resident budgeting time matrix  $\mathbf{R}=
[\lambda_{ij}]_{2\times2}$ for $i,j=1,2$ as in \cite{lee2015role}. Here $\lambda_{ij}\doteq \alpha_{ij}+\beta_{ji}$, being $\alpha_{ij}$ the probability of a human from patch $i$ is visiting the patch $j$ and $\beta_{ji}$ the probability of a mosquito from patch $j$,   is visiting the patch $i$. Some authors prefer not to consider the mobility of mosquitoes due to yours short life cycle (less than two weeks without captivity), in which case we assume  $\beta_{ji}=0$. Each $\lambda_{ij}$ is a constant in $[0,1]$ and $\sum_{j=1}^2\lambda_{ij}=1$ for $i=1,2$.
In this model we include bi--directional motion as in \cite{lee2015role}, that is,   a susceptible human (mosquito) in patch $i$ can be infected by an infected mosquito (human) from patch $i$ as well as by an infected mosquito (human) from patch $j$ who is visiting the patch $i$. Thus, the dynamic in two patches are represented
through the following  system of nonlinear  ODEs:

\begin{equation}\label{modelocomp}
\left\{\begin{array}{ll}
        &\dot{S_{h_i}}  = \Lambda_{h_i}+\omega_iR_{h_i}-S_{h_i}\sum_{j=1}^2\lambda_{ij}\beta_{hj}\epsilon_j\frac{I_{v_j}}{N_{h_j}}-\mu_{h_i}S_{h_i} \\ \\
  &\dot{I_{h_i}} = \sum_{j=1}^2\lambda_{ij}\beta_{hj}\epsilon_j\frac{I_{v_j}}{N_{h_j}}-\xi_{1i}\theta_{1i}(1-q_{1i})I_{h_i}-(\delta_i+\rho_i+\mu_{h_i})I_{h_i} \\ \\
  &\dot{R_{h_i}} = \xi_{1i}\theta_{1i}(1-q_{1i})I_{h_i}+\delta_iI_{h_i}-(\omega_i+\mu_{h_i})R_{h_i}  \\ \\
  &\dot{S_{v_i}} = \Lambda_{v_i}-S_{v_i}\sum_{j=1}^2\lambda_{ji}\beta_{vj}\epsilon_j\frac{I_{h_j}}{N_{h_j}}-\xi_{2i}\theta_{2i}(1-q_{2i})-\mu_{v_i}S_{v_i} \\ \\
   &\dot{I_{v_i}} = S_{v_i}\sum_{j=1}^2\lambda_{ji}\beta_{vj}\epsilon_j\frac{I_{h_j}}{N_{h_j}}-\xi_{2i}\theta_{2i}(1-q_{2i})-\mu_{v_i}I_{v_i},
    \quad \text{for} \quad i=1,2\\ \\
    &(\mathbf{N}_{h_1}(0), \mathbf{N}_{v_1}(0), \mathbf{N}_{h_2}(0), \mathbf{N}_{v_2}(0))= (\mathbf{N}_{h_1}^0, \mathbf{N}_{v_1}^0, \mathbf{N}_{h_2}^0,\mathbf{N}_{v_2}^0),
       \end{array}
\right.
\end{equation}
\noindent
where $(\mathbf{N}_{h_1}(0), \mathbf{N}_{v_1}(0), \mathbf{N}_{h_2}(0), \mathbf{N}_{v_2}(0))$ denotes an initial condition.
Let us define $N_H(t)=N_{h_1}(t)+N_{h_2}(t)$,  $N_V(t)=N_{v_1}(t)+N_{v_2}(t)$ and

\begin{equation}\label{lambdas}
    \begin{array}{ll}
       \Lambda_H=2\max\{\Lambda_{h_1},\;\Lambda_{h_2}\}, & \mu_H=\min\{\mu_{h_1}, \; \mu_{h_2}\} \\
       \Lambda_V=2\max\{\Lambda_{v_1},\;\Lambda_{v_2}\}, & \mu_V=\min\{\mu_{v_1},\;\mu_{v_2}\}.
    \end{array}
\end{equation}
\noindent
A  set of biological interest for the solutions of the system (\ref{modelocomp}) is
\begin{equation}\label{omegabar}
    \bar \Omega=\left\{ (\mathbf{N}_{h_1}, \mathbf{N}_{v_1}, \mathbf{N}_{h_2}, \mathbf{N}_{v_2}) \in \mathbb{R}_{10}^+: \; N_H\leq \frac{\Lambda_H}{\mu_H} \; \text{and} \, N_V\leq \frac{\Lambda_V}{\mu_V} \right\}.
\end{equation}
\noindent
The proof of invariance of $\bar \Omega$ can be be made using the results of  Lemma \ref{teoinvarianzaone}.
\subsection{Global basic reproductive number and numerical experiments}
In this subsection, we first compute the global basic reproductive number associated to the system (\ref{modelocomp}). Then,  we  obtain numerical experiments to generate an application of the mathematical model (\ref{modelocomp}) using  data from \cite{romero2018optimal}.  Let us denote as $\mathbf{E}_0=\left(\bar N_{h_1},0,0,\bar N_{v_1},0,\bar N_{h_2} ,0,0,\bar N_{v_2},0  \right)$ with
\begin{equation}\label{barN}
    \bar N_{h_i}=\frac{\Lambda_{h_i}}{\mu_{h_i}} \; \text{and} \; \bar N_{v_i}=\frac{\Lambda_{v_i}}{\xi_{2i}\theta_{2i}(1-q_{2i})+\mu_{v_i}} \; \text{for} \; i=1,2,
\end{equation}
to the DFE associated to the system (\ref{modelocomp}). Using a similar procedure to that Subsection \ref{section_reproductivo_one} with
\begin{equation*}
\mathbf{F}=\left(
    \begin{array}{cccc}
      0 & \lambda_{11}\beta_{h_1}\epsilon_1 & 0 & \lambda_{12}\beta_{h_2}\epsilon_2\frac{\bar N_{h_1}}{\bar N_{h_2}} \\
      \lambda_{11}\beta_{v_1}\epsilon_1\frac{\bar N_{v_1}}{\bar N_{h_1}} & 0 & \lambda_{21}\beta_{v_1}\epsilon_2\frac{\bar N_{v_1}}{\bar N_{h_2}} & 0 \\
      0 & \lambda_{21}\beta_{h_1}\epsilon_1\frac{\bar N_{h_2}}{\bar N_{h_1}} & 0  & \lambda_{22}\beta_{h_2}\epsilon_2 \\
      \lambda_{12}\beta_{v_1}\epsilon_1\frac{\bar N_{v_2}}{\bar N_{h_1}} & 0 & \lambda_{22}\beta_{v_2}\epsilon_2\frac{\bar N_{v_2}}{\bar N_{h_2}} & 0 \\
    \end{array}
  \right) \; \text{and}
  \end{equation*}
  {\footnotesize
\begin{equation*}
\mathbf{V}=\left(
            \begin{array}{cccc}
              \xi_{11}\theta_{11}(1-q_{11})+\delta_1+\rho_1+\mu_{h_1} & 0 & 0 & 0 \\ \\
              0 & \xi_{21}\theta_{21}(1-q_{21})+\mu_{v_1} & 0 & 0 \\ \\
              0 & 0 & \xi_{12}\theta_{12}(1-q_{12})+\delta_2+\rho_2+\mu_{h_2} & 0 \\ \\
              0 & 0 & 0 & \xi_{22}\theta_{22}(1-q_{22})+\mu_{v_2} \\
            \end{array}
          \right),
\end{equation*}}
\noindent
we get the following expression to the global basic reproductive number

\begin{equation}\label{R02}
\mathcal{R}_0=\left(\frac{\eta_1+\sqrt{\eta_2}}{2}\right)^{1/2},
 \end{equation}
 where

 \begin{equation*}
 \begin{array}{ll}
  \eta_1 &= \frac{\beta_{h_1}\epsilon_1}{\xi_{21}\theta_{21}(1-q_{21})+\mu_{v_1}}\frac{\bar N_{v_1}}{\bar N_{h_1}}\left[\frac{\lambda_{21}^2\beta_{v_2}\epsilon_2}{\left(\xi_{12}\theta_{12}(1-q_{12})+\delta_2+\rho_2+\mu_{h_2} \right)}+\frac{\lambda_{11}^2\beta_{v_1}\epsilon_1}{\left(\xi_{11}\theta_{11}(1-q_{11})+\delta_1+\rho_1+\mu_{h_1} \right)}  \right] \\ \\
    &+ \frac{\beta_{h_2}\epsilon_2}{\xi_{22}\theta_{22}(1-q_{22})+\mu_{v_2}}\frac{\bar N_{v_2}}{\bar N_{h_2}}\left[\frac{\lambda_{12}^2\beta_{v_1}\epsilon_1}{\left(\xi_{11}\theta_{11}(1-q_{11})+\delta_1+\rho_1+\mu_{h_1} \right)}+\frac{\lambda_{22}^2\beta_{v_2}\epsilon_2}{\left(\xi_{12}\theta_{12}(1-q_{12})+\delta_2+\rho_2+\mu_{h_2} \right)}  \right] \\ \\
    \eta_2&= \left[\frac{\lambda_{11}^2\beta_{v_1}\beta_{h_1}\epsilon_1^2}{\left(\xi_{21}\theta_{21}(1-q_{21})+\mu_{v_1}\right)\left(\xi_{11}\theta_{11}(1-q_{11})+\delta_1+\rho_1+\mu_{h_1} \right)}\frac{\bar N_{v_1}}{\bar N_{h_1}}\right.
    \left.+\frac{\lambda_{12}^2\beta_{v_2}\beta_{h_2}\epsilon_2^2}{\left( \xi_{22}\theta_{22}(1-q_{22})+\mu_{v_2} \right)\left(\xi_{12}\theta_{12}(1-q_{12})+\delta_2+\rho_2+\mu_{h_2} \right)}\frac{\bar N_{v_2}}{\bar N_{h_2}}  \right]^2\\ \\
   & +\left[\frac{\lambda_{21}^2\beta_{v_1}\beta_{h_1}\epsilon_1\epsilon_2}{\left(
\xi_{21}\theta_{21}(1-q_{21})+\mu_{v_1} \right)\left(\xi_{11}\theta_{11}(1-q_{11})+\delta_1+\rho_1+\mu_{h_1} \right)}\frac{\bar N_{v_1}}{\bar N_{h_1}} \right.
\left. +\frac{\lambda_{22}^2\beta_{v_2}\beta_{h_2}\epsilon_2^2}{\left( \xi_{22}\theta_{22}(1-q_{22})+\mu_{v_2} \right)\left(\xi_{12}\theta_{12}(1-q_{12})+\delta_2+\rho_2+\mu_{h_2} \right)}\frac{\bar N_{v_2}}{\bar N_{h_2}} \right]^2 \\ \\
    & +2\frac{\beta_{v_1}\beta_{v_2}\epsilon_1\epsilon_2}{\left(\xi_{11}\theta_{11}(1-q_{11})+\delta_1+\rho_1+\mu_{h_1} \right)\left(\xi_{12}\theta_{12}(1-q_{12})+\delta_2+\rho_2+\mu_{h_2}\right)}
    \times \left[ \frac{\lambda_{12}^2\lambda_{22}^2\beta_{h_2}^2\epsilon_2^2\bar N_{v_2}^2}{\left( \xi_{22}\theta_{22}(1-q_{22})+\mu_{v_2} \right)^2\bar N_{h_2}^2}+\frac{\lambda_{11}^2\lambda_{21}^2\beta_{h_1}^2\epsilon_1^2 \bar N_{v_1}^2}{\left(
\xi_{21}\theta_{21}(1-q_{21})+\mu_{v_1}\right)^2\bar N_{h_1}^2}  \right] \\ \\
    &+ 2\frac{\beta_{h_1}\beta_{h_2}\beta_{v_1}\beta_{v_2}\epsilon_1^2\epsilon_2^2}
    {\left(\xi_{21}\theta_{21}(1-q_{21})+\mu_{v_1} \right)\left(\xi_{11}\theta_{11}(1-q_{11})+\delta_1+\rho_1+\mu_{h_1}\right)\left(\xi_{22}\theta_{22}(1-q_{22})+\mu_{v_2} \right)\left(\xi_{12}\theta_{12}(1-q_{12})+\delta_2+\rho_2+\mu_{h_2} \right)} \\ \\
    &\times
    \frac{\bar N_{v_1}\bar N_{v_2}}{\bar N_{h_1}\bar N_{h_2}}\left[4\lambda_{11}\lambda_{12}\lambda_{21}\lambda_{22}-\lambda_{11}^2\lambda_{22}^2-\lambda_{12}^2\lambda_{21}^2  \right].
 \end{array}
 \end{equation*}
\noindent
Considering the uncopling system (that is, $\lambda_{11}=1$ and $\lambda_{22}=1$) in (\ref{R02}), we obtain the local basic reproductive number for each patch given in  (\ref{R0one}).
\par
In what follows, we make some  numerical experiments. For this purpose, we are going to consider the following hypothesis: (a) the patch 1 and patch 2 represent  rural  areas (RA) and  urban areas (UA) from the municipality of Tumaco (Colombia) as in \cite{romero2018optimal,leiton2018analisis}, respectively. (b)  The epidemiological outbreak begins in RA and the individuals in UA acquire the infection due to the coupling between the two patches. Therefore (unless otherwise stated), the initial condition  will be $S_{h_1}(0)=100000$, $I_{h_1}(0)= 30000$, $R_{h_1}(0)=20000$, $S_{v_1}(0)=50000$, $I_{v_1}(0)= 10000$, $S_{h_2}(0)= 100000$,  $S_{v_2}(0)=50000$  and all others zero. (c) Mosquitos are fumigated only with cyflutrin (data from Table \ref{tabla_insecticidas}).  (d) The infected patient are treated with ACT (data from Table \ref{tabla_droga}). (d)  The  resistance acquisition ratio in RA is higher than UA due to in RA individuals  are continuously exposed to the parasite, that is, $q_{11}=0.1$, $q_{12}=0.09$, $q_{21}= 0.05$ and $q_{22}=0.04$.
Besides, we will consider the following coupling scenarios poposed by Lee et al. in \cite{lee2015role}:
\begin{itemize}
  \item [(S1)] Uncoupled: when there are no visits between patches, that is, $\lambda_{11}=\lambda_{22}=1$ and others are equal to zero
  \item [(S2)] Weakly--coupling:  small values for $\lambda_{12}$ and $\lambda_{21}$.
  \item [(S3)] Strongly--coupling: when visitors from patch 2 spend quite an amount of time in patch 1, that is, $\lambda_{22}<\lambda_{11}$.
\end{itemize}
\noindent
Table \ref{tabla3} shows the values of the parameters in the residence--time matrix considering  different scenarios of coupling.

\begin{table}[H]
\begin{center}
\caption{ Values of the parameters in the residence--time matrix.}
\label{tabla3}
\begin{tabular}{lll}\hline
         Scenario             &         Values of the parameters                    \\ \hline
  Uncoupled           & $\lambda_{11}=\lambda_{22}=1$,   $\lambda_{12}=\lambda_{21}=0$  \\
 Weakly--coupling       & $\lambda_{11}=\lambda_{22}=0.9$, $\lambda_{12}=\lambda_{21}=0.1$ \\
 Strongly--coupling     & $\lambda_{11}=\lambda_{22}=0.4$, $\lambda_{12}=\lambda_{21}=0.6$  \\ \hline
\end{tabular}
\end{center}
\end{table}

Figure \ref{desacoplado} shows the behavior of the solutions  when the system (\ref{modelocomp}) is uncoupled.  If the disease begins in patch 1, the disease does not spread to patch 2.

\begin{figure}[htbp]
\centering
\includegraphics[width=7 cm, height=7.0 cm]{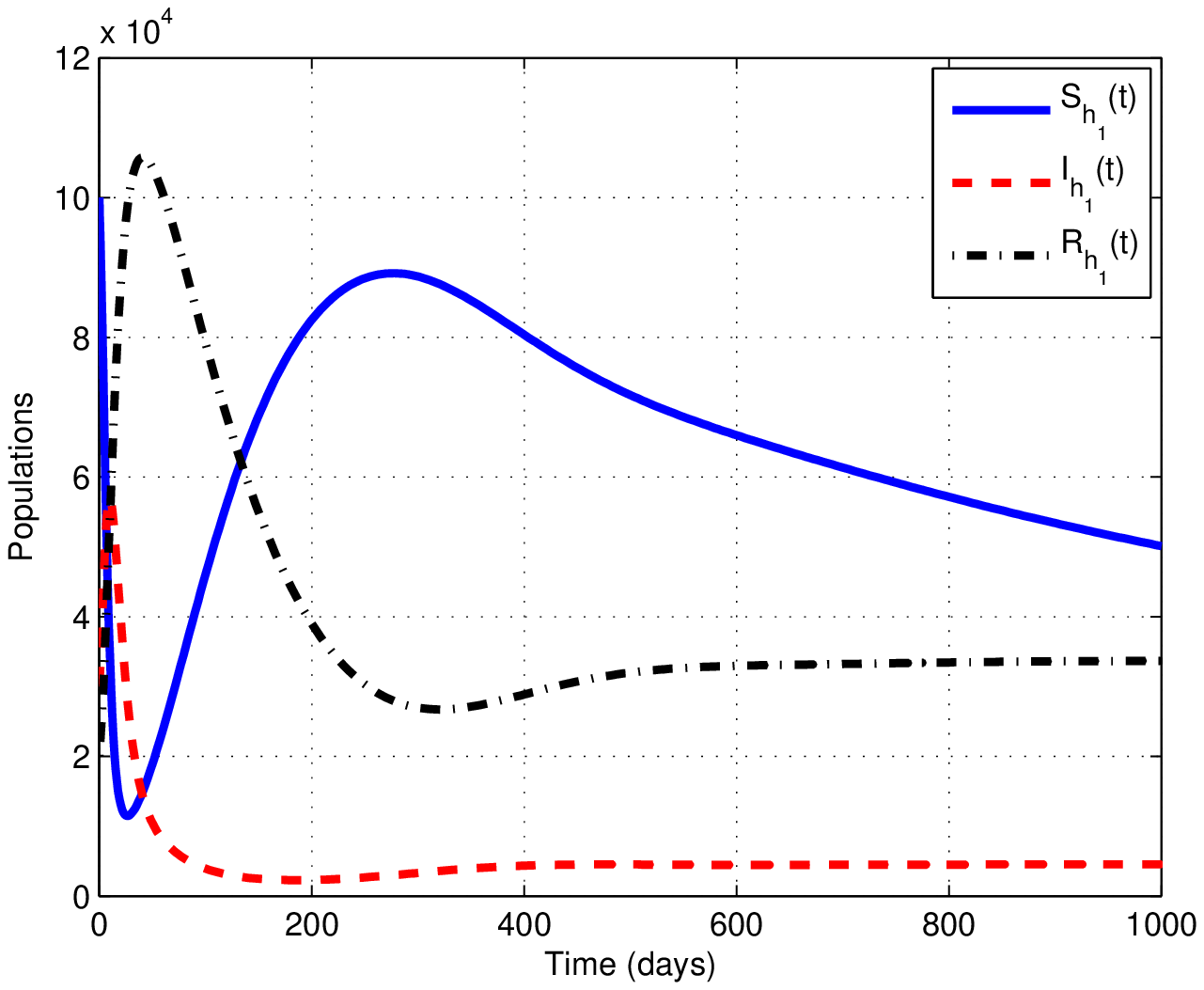}
\includegraphics[width=7 cm, height=7 cm]{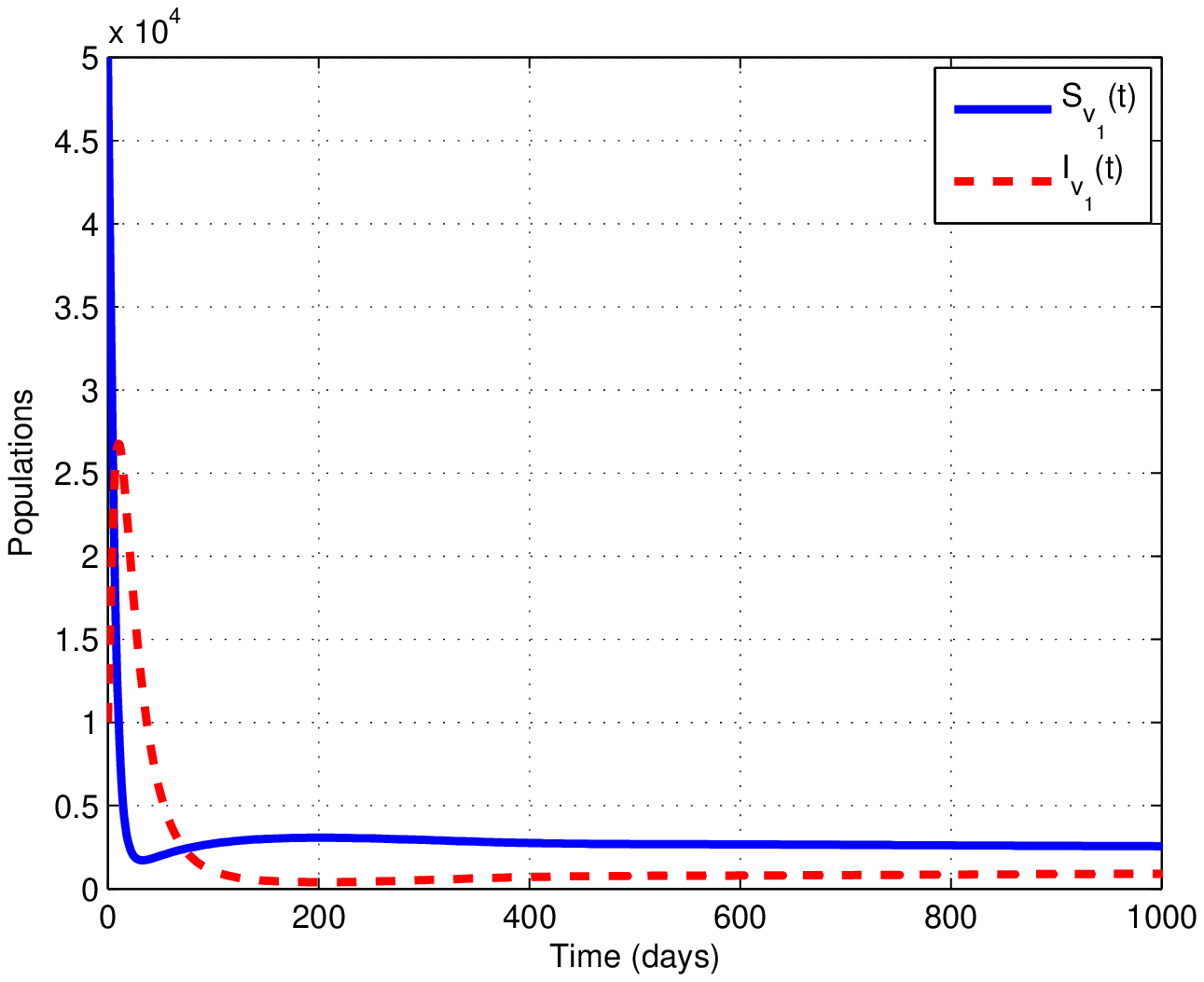} \\
\includegraphics[width=7 cm, height=7.0 cm]{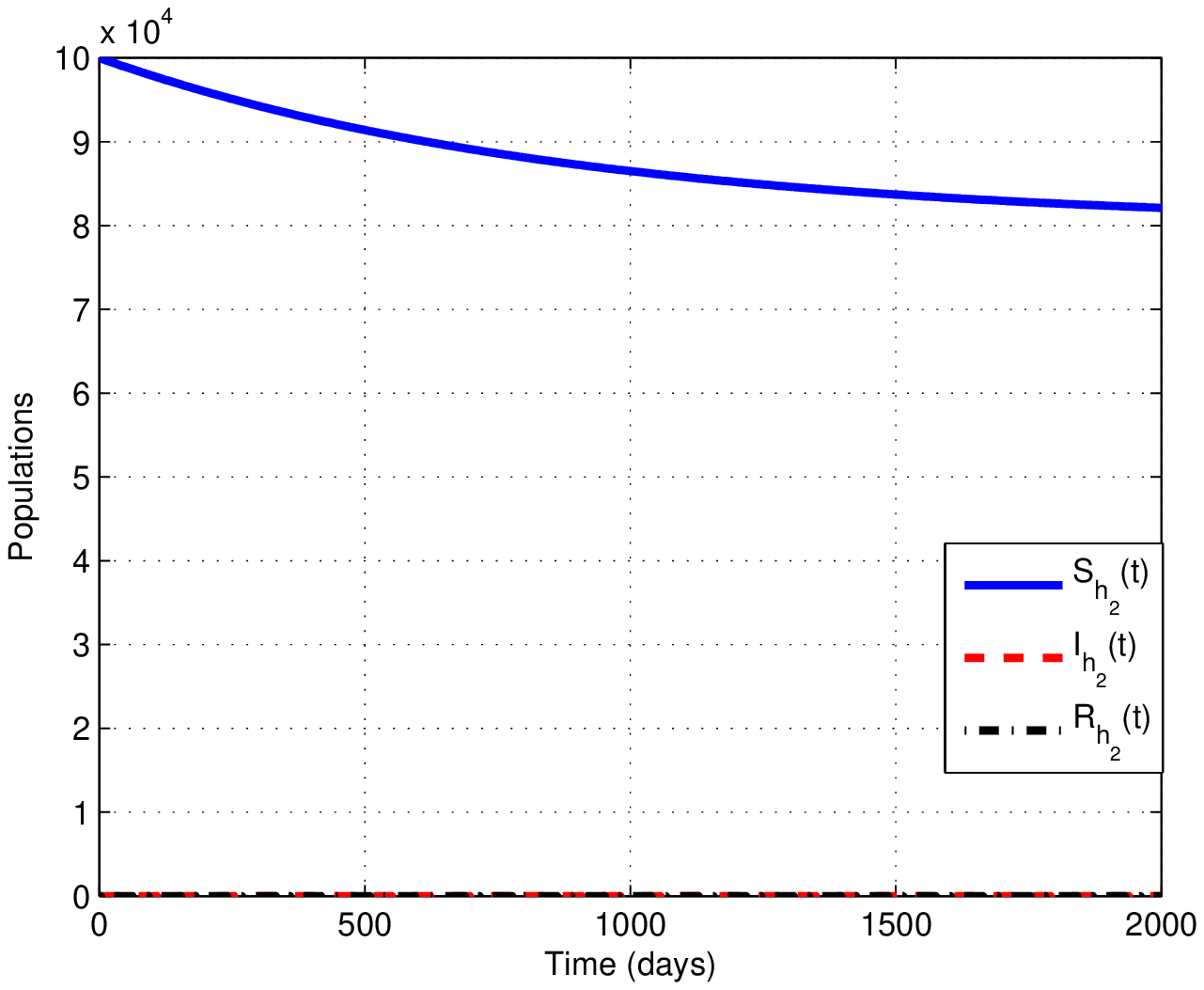}
\includegraphics[width=7 cm, height=7.0 cm]{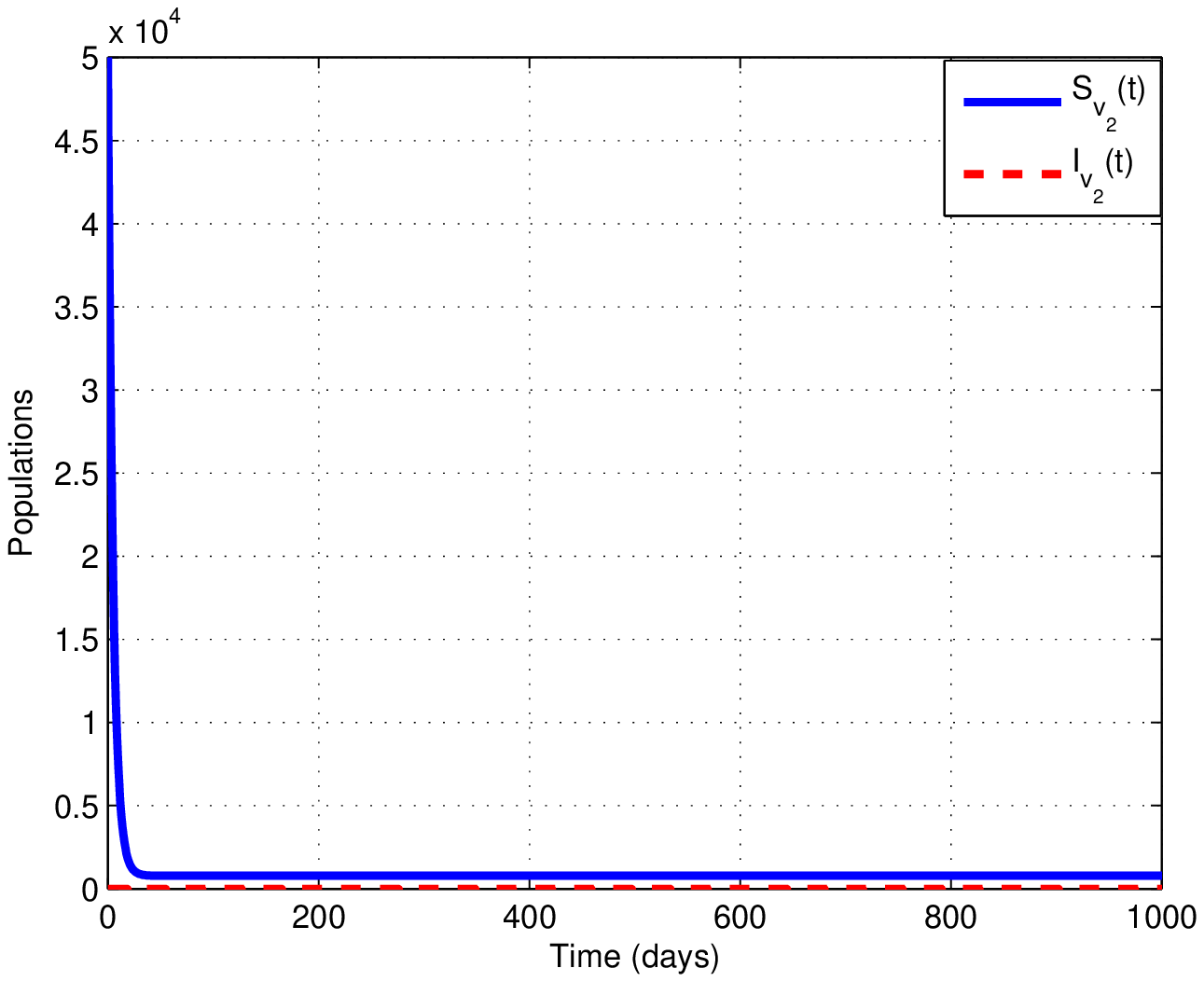}
\caption{ Numerical simulations of uncoupled system (\ref{modelocomp}) using data from \cite{romero2018optimal}  ($\lambda_{11}=\lambda_{22}=1$) .  The  initial condition is (100000, 30000, 20000,  50000, 10000, 100000, 0, 0, 50000, 0).  Here $\mathcal{R}_{0_1}=2.52$, $\mathcal{R}_{0_2}=0.12$ and $\mathcal{\mathcal{R}}_0=2.15$.}
\label{desacoplado}
\end{figure}

Figure \ref{weakly} shows the behavior of the humans and mosquitoes populations in patches 1 and  2, respectively, considering weakly--coupling.  Here,  the disease is spread from patch 1  to patch 2 during the first 50 days, then the disease is eliminated in patch 2, and remains at low load in patch 1.
\begin{figure}[htbp]
\centering
\includegraphics[width=7 cm, height=7 cm]{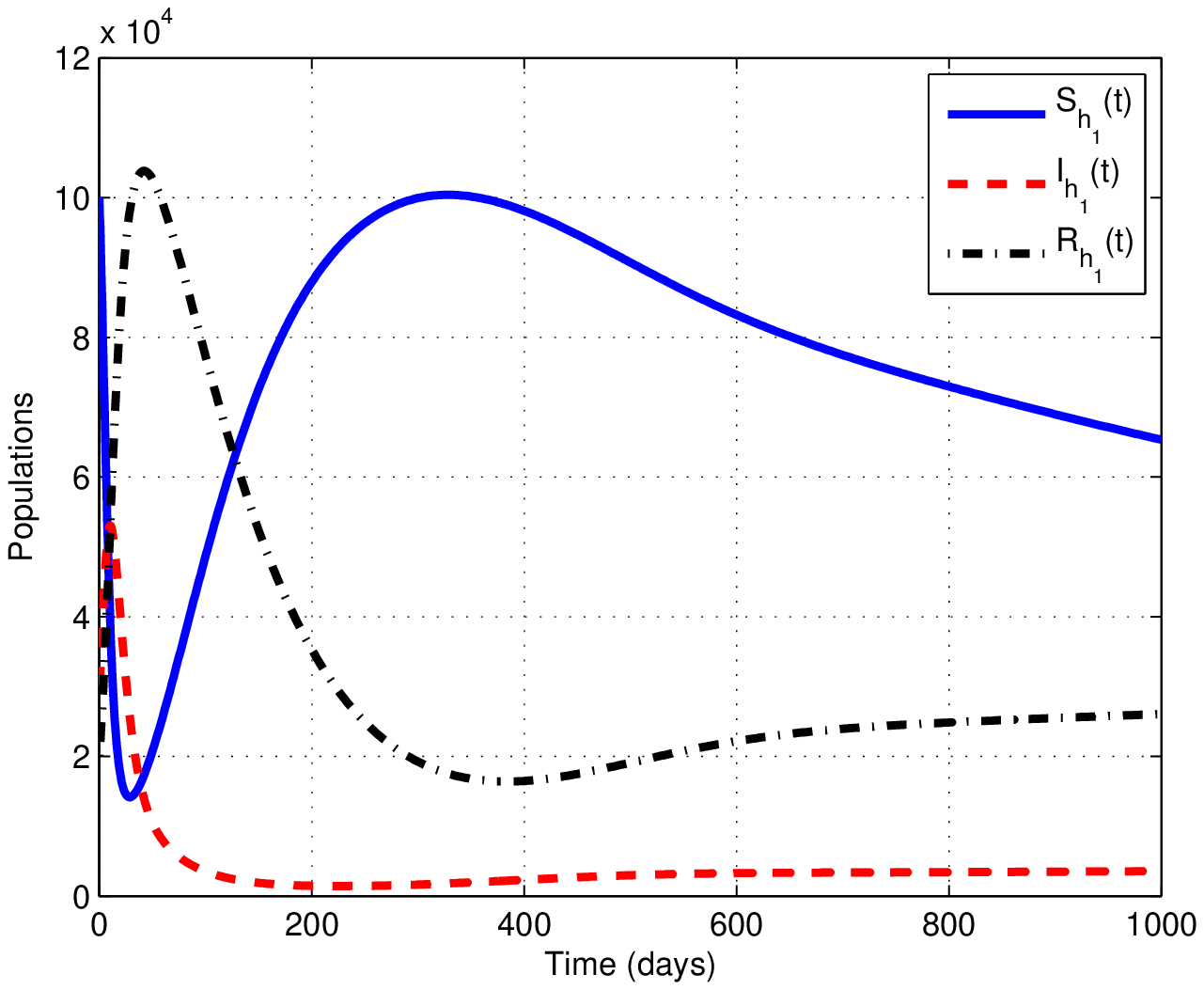}
\includegraphics[width=7 cm, height=7.0 cm]{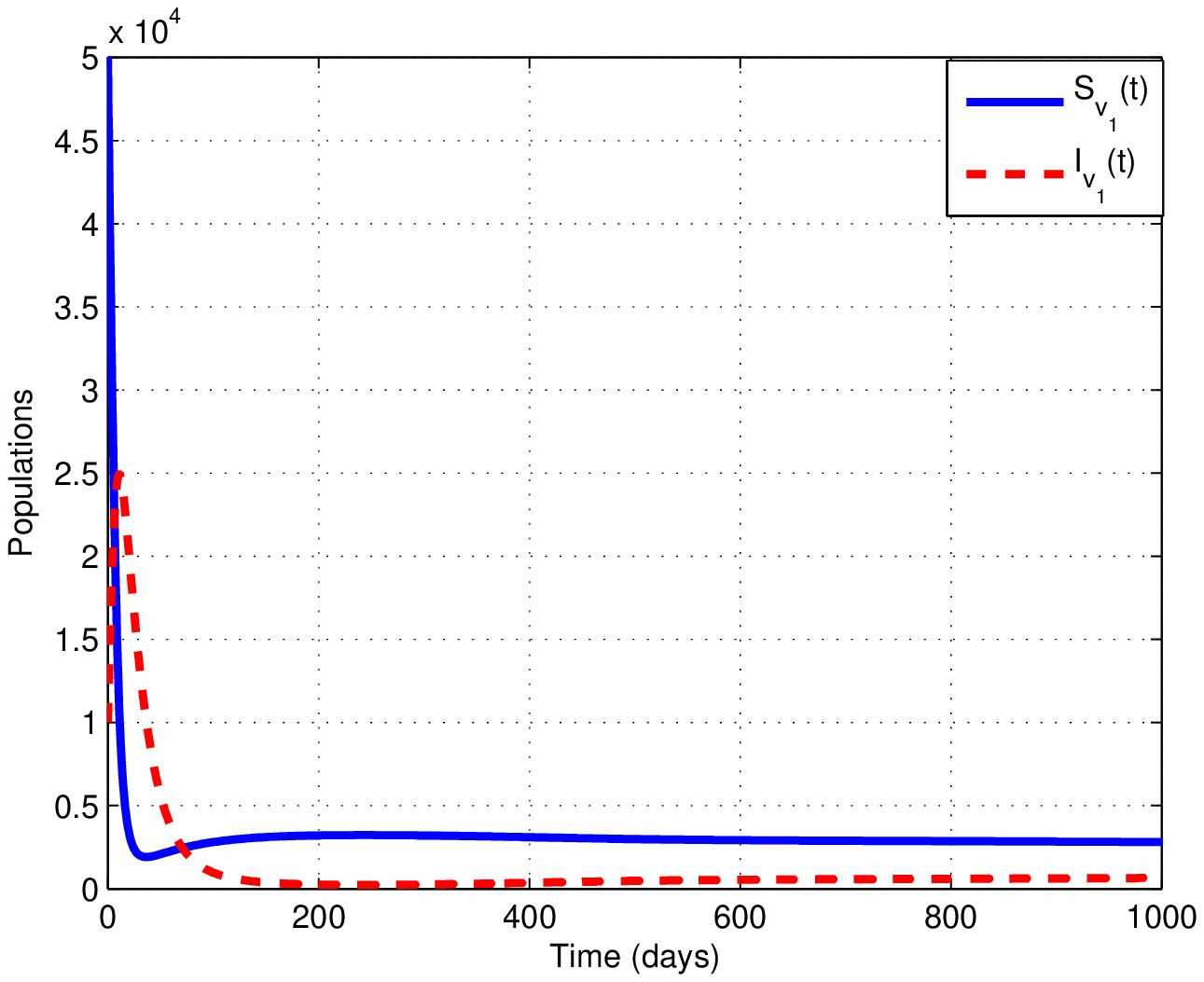} \\
\includegraphics[width=7 cm, height=7.0 cm]{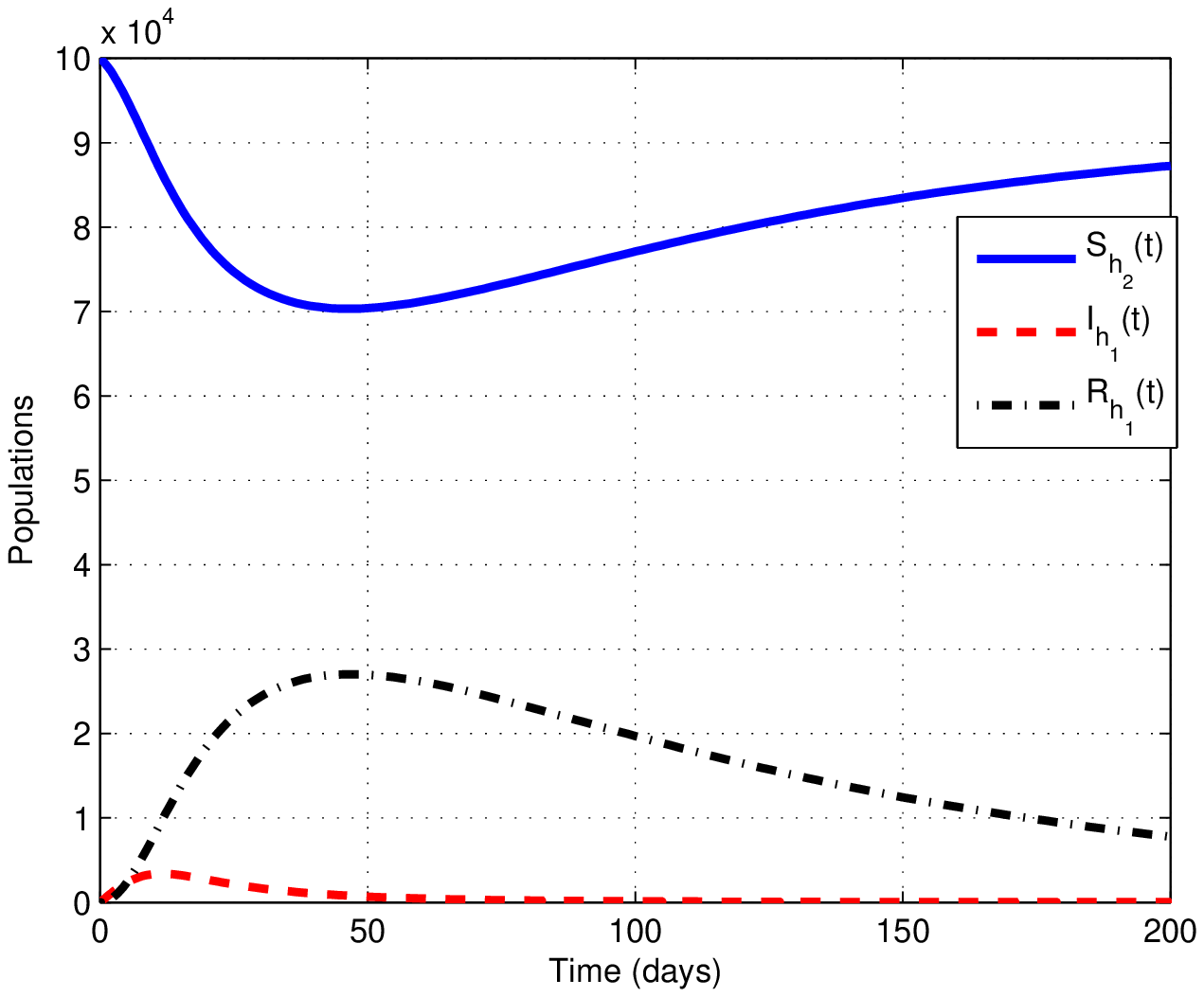}
\includegraphics[width=7 cm, height=7.0 cm]{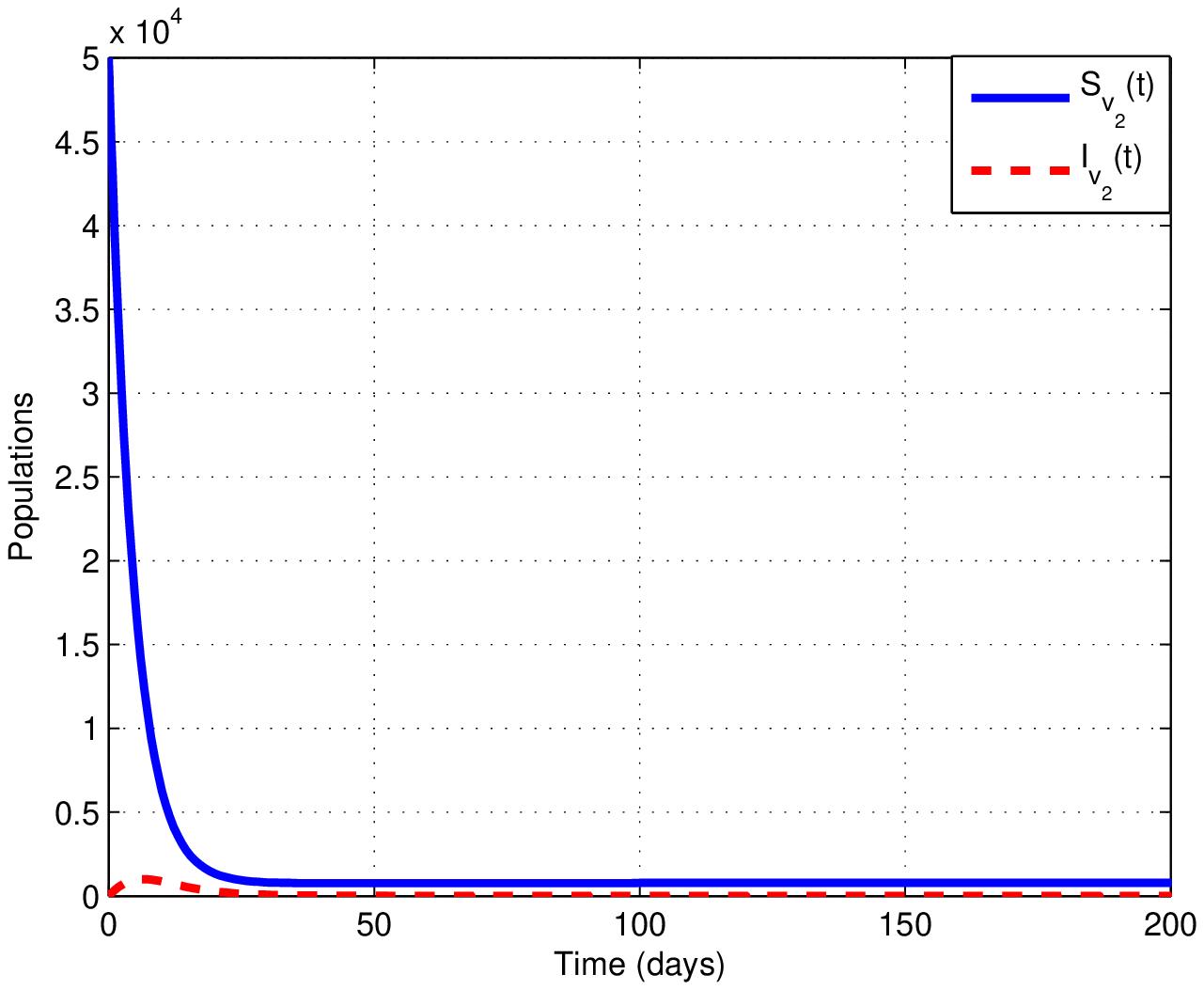}
\caption{Numerical simulations of weakly--coupled system (\ref{modelocomp}) using data from \cite{romero2018optimal} ($\lambda_{11}=\lambda_{22}=0.9$ and $\lambda_{12}=\lambda_{21}=0.1$). The  initial condition is (100000, 30000, 20000,  50000, 10000, 100000, 0, 0, 50000, 0).  Here $\mathcal{R}_{0_1}=1.46$, $\mathcal{R}_{0_2}=0.6$ and $\mathcal{\mathcal{R}}_0=3.47$.}
\label{weakly}
\end{figure}

The strongly--coupling scenario is illustrated in Figure \ref{strong}. Here, the disease is spread from patch 1  to patch 2 during the first 100 days and the infection persists in both patches. After 100 days, the disease is eliminated in both patches.

\begin{figure}[htbp]
\centering
\includegraphics[width=7 cm, height=7.0 cm]{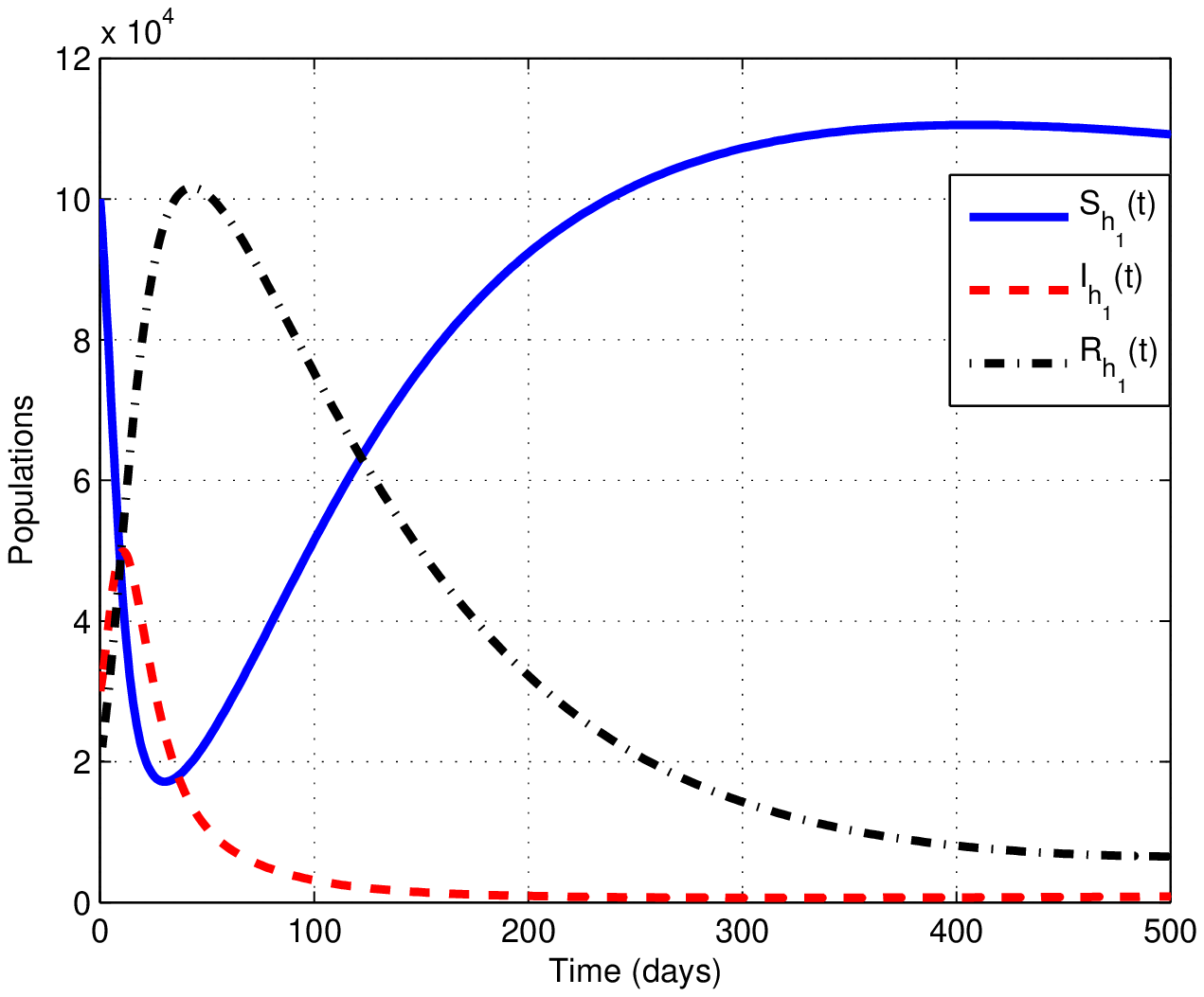}
\includegraphics[width=7 cm, height=7.0 cm]{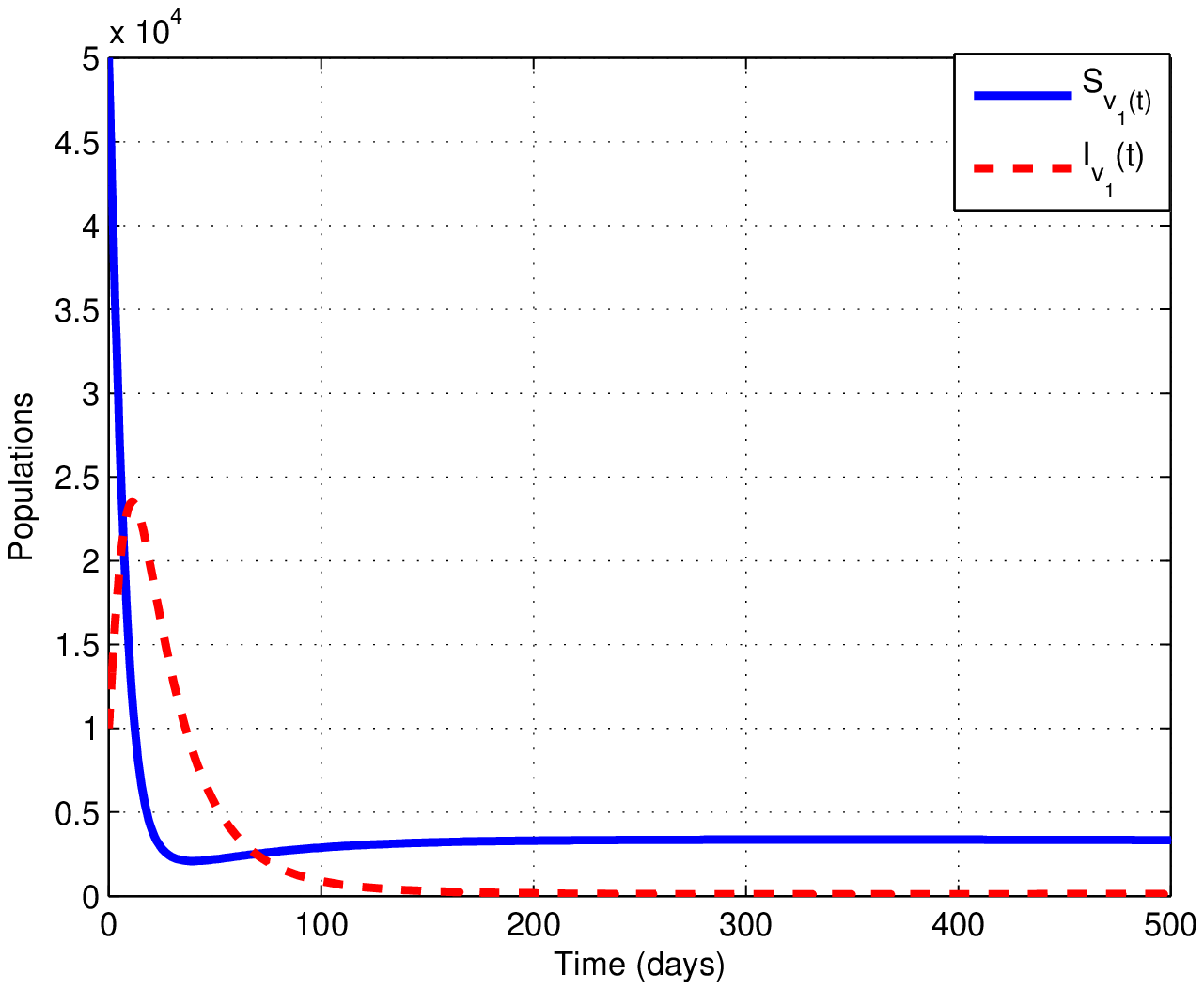} \\
\includegraphics[width=7 cm, height=7.0 cm]{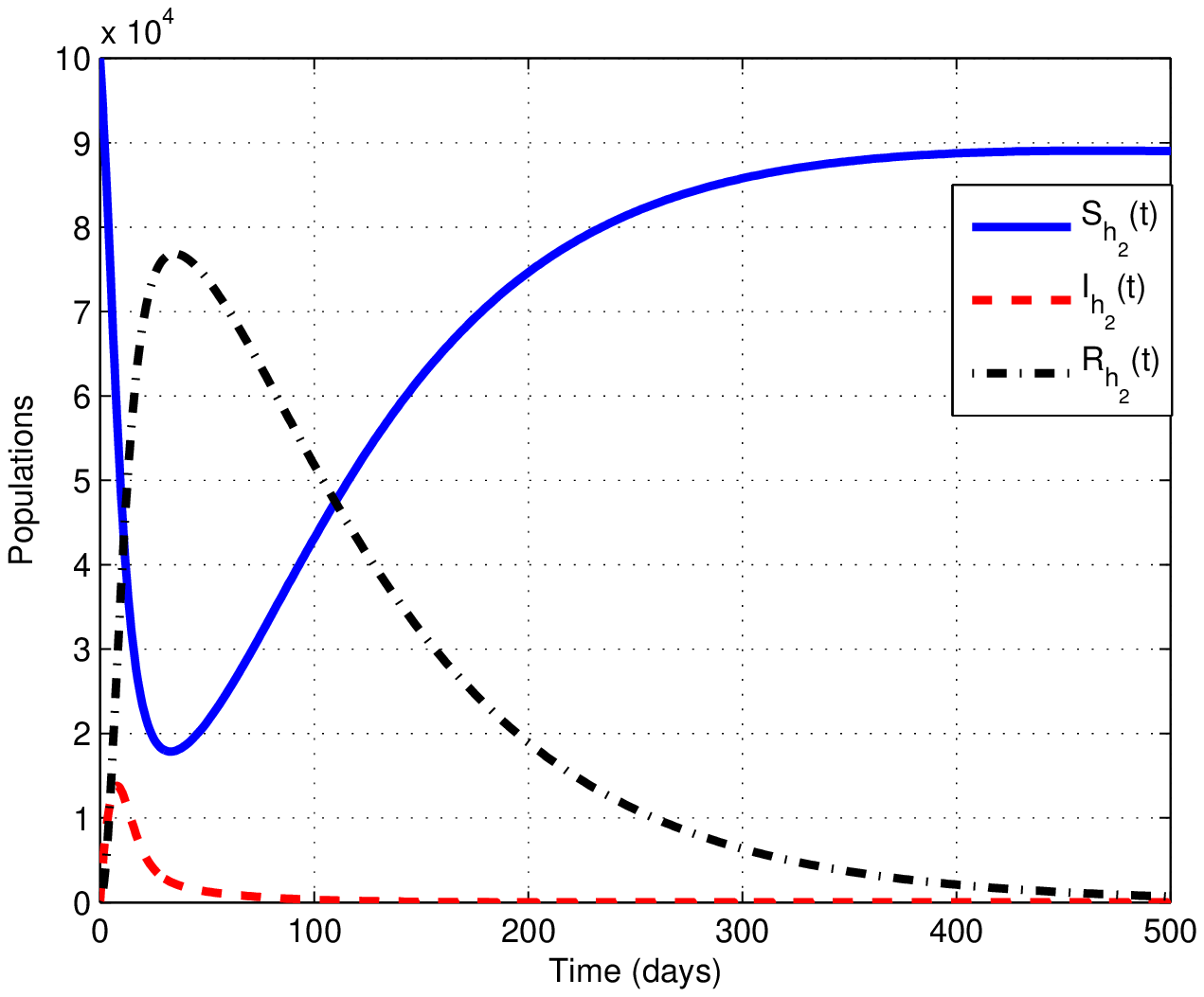}
\includegraphics[width=7 cm, height=7.0 cm]{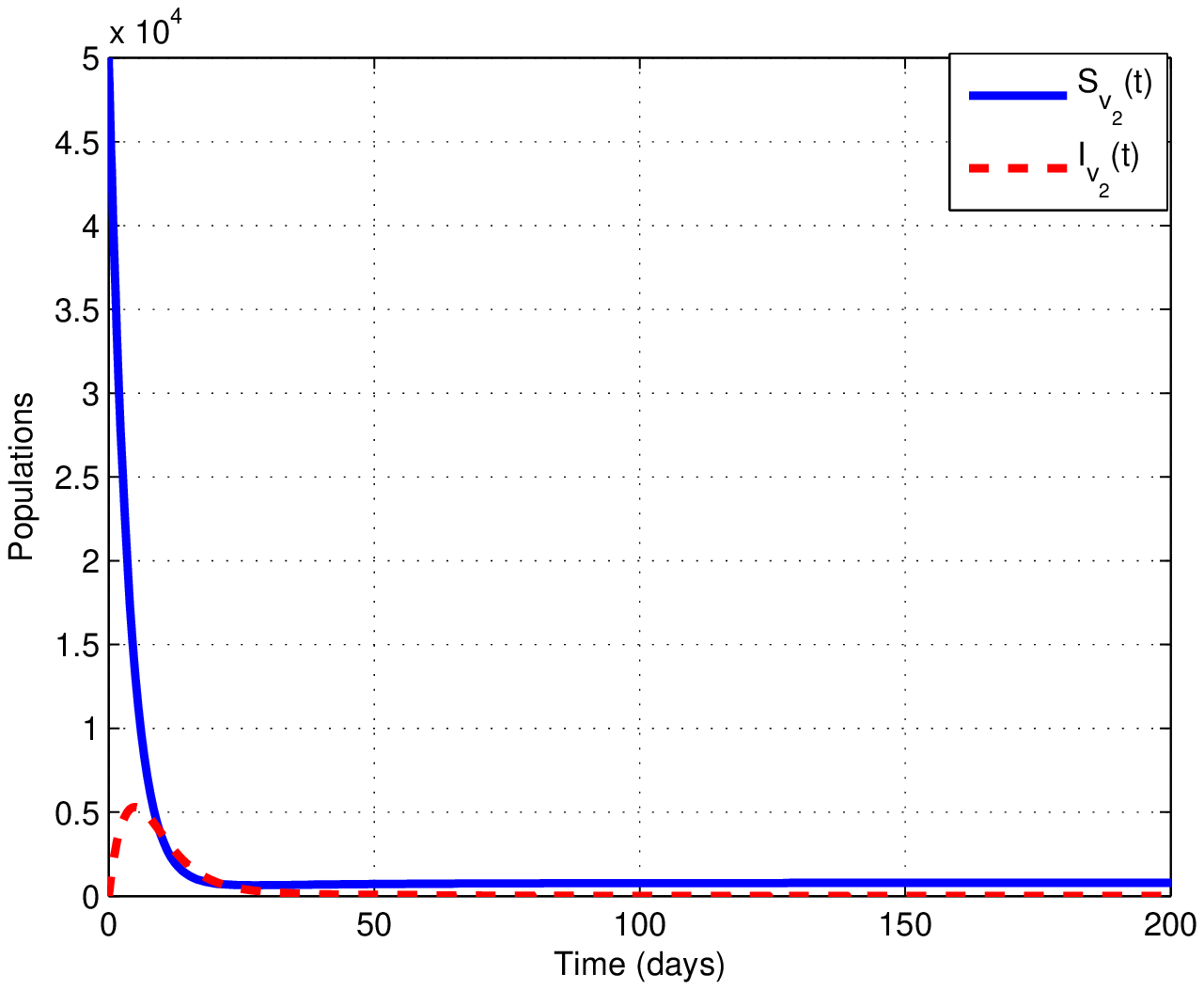}
\caption{ Numerical simulations of strongly--coupled system (\ref{modelocomp}) using data from \cite{romero2018optimal} ($\lambda_{11}=\lambda_{22}=0.4$ and $\lambda_{12}=\lambda_{21}=0.6$).  The  initial condition is (100000, 30000, 20000,  50000, 10000, 100000, 0, 0, 50000, 0). Here,  $\mathcal{R}_{0_1}=1.01$, $\mathcal{R}_{0_2}=0.9$ and $\mathcal{\mathcal{R}}_0=2.9$.}
\label{strong}
\end{figure}
\subsection{Local sensitivity analysis of parameters}
In this subsection we determine the sensitivity indices of the parameters to the $\mathcal{R}_0$, considering strongly-- coupling and data from  \cite{romero2018optimal} .  The sensitivity indices are computed through  the \textit{normalized forward sensitivity index} \cite{chitnis2008determining},  which allow us to measure the relative change of the variable $\mathcal{R}_0$ when a parameter changes. When the variable is a differentiable function of the parameter, the sensitivity index may be alternatively defined using partial derivatives \cite{chitnis2008determining}.  If we denote the variable as $u$ which depends on a parameter  $p$, the sensitivity index is defined by
\begin{equation}\label{indice}
    \Gamma_p^u\doteq \frac{\partial u}{\partial p}\frac{p}{u}.
\end{equation}
\noindent
Given the explicit formula for $\mathcal{R}_0$ in (\ref{R02}), we determine an analytical expression for the
sensitivity indices of $\mathcal{R}_0$ with respect to each parameter that comprise it.  In Table \ref{tabla3valoresindices} we show the values of the sensitivity indices, where P1 and P2 mean patch 1 and patch 2, respectively.

\begin{table}[H]
\begin{center}
{\footnotesize
\caption{ Sensitivity indices to the $\mathcal{R}_0$ with respect to parameters.}\label{tabla3valoresindices}
\begin{tabular}{lll|lll}
\hline
  Parameter &  Index P1 & Index P2  & Parameter &  Index P1 & Index P2          \\ \hline
 $\Lambda_{h_i}$   & -0.033     & -0.4     &       $\lambda_{11}$   & 0.00012                 &0.0086                  \\
$\omega_i$       & 0  &  0    &       $\lambda_{12}$   & 0.0042                  &0.0042                          \\
  $\beta_{h_i}$   & 0.01     & 0.50         &       $\lambda_{22}$   & 0.0030                  &0.0030                              \\
  $\beta_{v_i}$   & 0.09     & 0.49                & $\lambda_{21}$   &0.3887                  &0.3887                          \\
  $\epsilon_i$     &0.09                 & 0.90                 &$q_{1i}$         & -0.003                &-0.0014           \\
  $\mu_{h_i}$      & 0.012             &-0.49                & $q_{2i}$         & -0.1164                 &   -0.00058            \\
  $\mu_{v_i}$      & -0.011                & -0.012              &$\xi_{1i}$       & -0.3488                         &-0.1405                    \\
  $\rho_i$         & -0.0053               & -0.0082       &           $\xi_{2i}$       &  0.0034                       &  0.0045                 \\
  $\delta_i$       &-0.0033                &-0.0027             & $\theta_{1i}$    &   -0.54                       & -0.65               \\
  $\Lambda_{v_i}$ & 0.0133                 & 0.50               & $\theta_{2i}$    &   -0.45                       &  -0.98           \\ \hline
\end{tabular}}
 \end{center}
\end{table}

From Table \ref{tabla3valoresindices}, in both rural (patch 1) and urban (patch 2) areas, $ \mathcal{R}_0 $ is more sensitive to the parameters corresponding to recovery rate due to the drug $\theta_{1i}$ with $i=1,2$ and death rate due to the insecticides $\theta_{2i}$ with $i=1,2$. An interpretation of these indices is given as follows: in RA, given that $\Gamma_{\theta_{11}}=-0.54$,  increasing (or decreasing) $\theta_{11}$ in  10\% implies  that $\mathcal{R}_0$ decreases (or increases) in 5.4\%. An analogous reasoning can be made for  the others sensitivity indices.
The information provided by the sensitivity indices to the $\mathcal{R}_0$, will be used in the next section, in which we will propose some control strategies for the malaria disease.


\section{Optimal control problem}\label{Seccion:optimalcontrol}
In this section an optimal control problem applied to the model (\ref {modelocomp}) is formulated. Here,  we are going to consider that the parameters corresponding to recovery rate due to the drug and death rate due to the insecticides  $\theta_{ij}$ with $i,j=1,2$ will be the controls, therefore they will be functions depending on time. The first objective will be to minimize a performance index or cost function by the use of drugs and insecticides. For this purpose, we assume that $\theta_{i1}$ with $i=1,2$ and $\theta_{j2}$ with $j=1,2$ are the controls by drugs and insecticides, respectively,  which assume values between $ 0 $ and $ 1 $, where $ \theta_{ij} = 0 $ is assumed if the use of drugs (or insecticides) is ineffective and $ \theta_{ij} = 1 $ if the use of drugs (or insecticides) is completely effective, that is, all individuals recover with medication and all mosquitoes die with insecticides. In this sense, for $i$ and $j$ fixed,  the control variable $\theta_{ij} (t) $ provides information about  amount of drug or insecticides  that must be supplied at time $ t $.
\par
The second objective will be to minimize the number of infected humans and infected mosquitoes in  each patch. For this purpose, the following performance index or cost function is considered:

\begin{equation}\label{funciondecostomodelo1def}
 J[\bm{\theta}]=\int_{0}^{T}\left[c_1(I_{h_1}+I_{v_1})+c_2(I_{h_2}+I_{v_2})+\frac{1}{2}(d_1\theta_{11}^2+d_2\theta_{21}^2+d_3\theta_{12}^2+d_4\theta_{22}^2)  \right]dt,
\end{equation}
\noindent
where $\bm{\theta}=(\theta_{11}, \theta_{21}, \theta_{12}, \theta_{22})$ is the vector of controls, $ c_1 $ and $ c_2 $ represent social costs, which depend on the number of individuals with malaria and the number of mosquitoes with the parasite, and $\frac{1}{2}(d_1\theta_{11}^2+d_2\theta_{21}^2+d_3\theta_{12}^2+d_4\theta_{22}^2) $  defines the absolute costs associated with the control strategies, such as, implementation, ordering, distribution, marketing, among others. For calculation purposes, we will denote to the integrand of the performance index given on (\ref{funciondecostomodelo1def}) as
\begin{equation}\label{f0}
    f_0(t, \mathbf{X}, \bm {\theta})=c_1(I_{h_1}+I_{v_1})+c_2(I_{h_2}+I_{v_2})+\frac{1}{2}(d_1\theta_{11}^2+d_2\theta_{21}^2+d_3\theta_{12}^2+d_4\theta_{22}^2)  ,
\end{equation}
\noindent where $\mathbf{X}$ represents the vector of states.
\par
With the above considerations, the following control problem is formulated.
\begin{equation}\label{modelocontrolcomp}
  \left\{ \begin{array}{ll}
  &\min_{\bm{\theta}\in U} J[\bm{\theta}]=\int \limits_{0}^{T}f_0(t. \mathbf{X}, \bm{\theta})dt\\ \\
          &\dot{S_{h_i}} = \Lambda_{h_i}+\omega_iR_{h_i}-S_{h_i}\sum_{j=1}^2\lambda_{ij}\beta_{hj}\epsilon_j\frac{I_{v_j}}{N_{h_j}}-\mu_{h_i}S_{h_i} \\ \\
  &\dot{I_{h_i}} = \sum_{j=1}^2\lambda_{ij}\beta_{hj}\epsilon_j\frac{I_{v_j}}{N_{h_j}}-\xi_{1i}\theta_{1i}(t)(1-q_{1i})I_{h_i}-(\delta_i+\rho_i+\mu_{h_i})I_{h_i} \\ \\
  &\dot{R_{h_i}} = \xi_{1i}\theta_{1i}(t)(1-q_{1i})I_{h_i}+\delta_i I_{h_i}-(\omega_i+\mu_{h_i})R_{h_i} \\ \\
  &\dot{S_{v_i}} = \Lambda_{v_i}-S_{v_i}\sum_{j=1}^2\lambda_{ji}\beta_{vj}\epsilon_j\frac{I_{h_j}}{N_{h_j}}-\xi_{2i}\theta_{2i}(t)(1-q_{2i})-\mu_{v_i}S_{v_i} \\  \\
   &\dot{I_{v_i}} = S_{v_i}\sum_{j=1}^2\lambda_{ji}\beta_{vj}\epsilon_j\frac{I_{h_j}}{N_{h_j}}-\xi_{2i}\theta_{2i}(t)(1-q_{2i})-\mu_{v_i}I_{v_i}, \quad \text{for} \quad i=1,2 \\ \\
   &\mathbf{X}(0) = ( \mathbf{N}_{h_1}(0),  \mathbf{N}_{v_1}(0),  \mathbf{N}_{h_2}(0),  \mathbf{N}_{v_2}(0))=\mathbf{X}_0 \\ \\
 &\mathbf{ X}(T) = (\mathbf{N}_{h_1}(T),  \mathbf{N}_{v_1}(T),  \mathbf{N}_{h_2}(T), \mathbf{N}_{v_2}(T))=\mathbf{X}_1.
          \end{array}
  \right.
\end{equation}

In above the formulation, we assume an initial time $ t_0 = 0 $, a final time $T$ fixed which represents the implementation time of the control strategies,  free dynamic variables $ \mathbf{X}_1$  in the final time, and the initial condition $ \mathbf{X}_0$  being a non--trivial equilibrium of the system (\ref {modelocomp}). Additionally, we assume that the  controls  are in a \textit{set of admissible controls} $\mathcal{U}$ which contains to all Lebesgue measurables  functions with values in the interval $[0,1]$ and $t\in[0,T]$.


\subsection{Existence of an optimal control}
In this section,  we use the classic existence theorem proposed by Lenhart and Workman  \cite{lenhart2007optimal} to prove the existence of an  optimal control $\bm {\theta}^*$ for the formulation (\ref{modelocontrolcomp}). Let
$U=[0,1]^4$ the set where  $\bm{\theta}$ assumes its values (set of controls), and $f(t,\mathbf{X},\bm\theta)$ the state equations of the right side of (\ref{modelocontrolcomp}). To guarantee the existence of optimal controls, hypotheses (H1) to (H5) from \cite{lenhart2007optimal} must be verified, that is,
\begin{itemize}\label{condicionescontrolpercapita}
\item [(H1)]
\begin{itemize}
        \item [(a)]$|f(t,\mathbf{0},\mathbf{0})| \leq C$
        \item [(b)] $|f_{\mathbf{X}}(t,\mathbf{X},\bm\theta)|  \leq C(1+|\bm\theta|)$
        \item [(c)] $|f_{\bm{\theta}}(t,\mathbf{X},\bm\theta)|\leq C$.
 \end{itemize}
  \item [(H2)] The set of controls   $U$ is  convex.
  \item [(H3)] $f(t,\mathbf{X},\bm\theta)=\alpha(t,\mathbf{X})+\beta(t,\mathbf{X})\bm\theta$.
  \item [(H4)] The integrand of the performance index $f_0(t,\mathbf{X},\bm\theta)$ defined in (\ref{f0}) is convex for $\bm\theta\in U$.
  \item [(H5)] $f_0(t,\mathbf{X},\bm\theta) \geq c_1 |\bm\theta|^b-c_2$ with $c_1>0$ and $b>1$.
\end{itemize}

We will proof the hypothesis (H1)(a) and (H5), since the others are obvious. For this purpose,  the following results are enunciated and proved.

\begin{lemma}\label{cota}
Let
\begin{equation}\label{xi}
    \frac{\xi}{2}=\{\max{\xi_{11}^2(1-q_{11})^2, \xi_{12}^2(1-q_{12})^2, \xi_{21}^2(1-q_{21})^2, \xi_{22}^2(1-q_{22})^2}\}.
\end{equation}
Then

\begin{equation}\label{eqcota}
 |f_{\bm\theta}(t,\mathbf{X},\bm\theta)|\leq \sqrt{\xi\left[\left(\frac{\Lambda_H}{\mu_H}  \right)^2+\left(\frac{\Lambda_V}{\mu_V}  \right)^2 \right]},
\end{equation}

where $\Lambda_H$, $\Lambda_V$, $\mu_H$ and $\mu_V$ are defined on (\ref{lambdas}), and $f_{\bm\theta}(t,\mathbf{X},\bm\theta)$ is the matrix obtained by differentiating  of the state equations of the right side of the system (\ref{modelocontrolcomp}) with respect to $\bm\theta$, whic is given by

\begin{equation}\label{fsubv}
    f_{\bm\theta}(t,\mathbf{X},\bm\theta)=\left(
                 \begin{array}{cccc}
                   0  &  0  & 0 & 0 \\
                -\xi_{11}(1-q_{11})I_{h_1}  & 0 & 0 & 0 \\
                   \xi_{11}(1-q_{11})I_{h_1} & 0 & 0 & 0 \\
                   0 & -\xi_{21}(1-q_{11})S_{v_1} & 0 & 0 \\
                   0 & -\xi_{21}(1-q_{11})S_{v_1} & 0 & 0 \\
                  0 & 0 & 0 & 0 \\
                 0 & 0 & -\xi_{12}(1-q_{12})I_{h_2} & 0 \\
                  0 &  0 &  \xi_{12}(1-q_{12})I_{h_2} & 0 \\
                  0 & 0 & 0 & -\xi_{22}(1-q_{22})S_{v_2} \\
                  0 & 0 & 0 & -\xi_{22}(1-q_{22})I_{v_2}
                 \end{array}
               \right).
\end{equation}
\end{lemma}
\begin{proof}
Computing the \textit{Euclidean norm} of matrix (\ref{fsubv}), we obtain
{\scriptsize
\begin{equation*}
    \begin{array}{ll}
    | f_{\bm\theta}(t,\mathbf{X},\bm\theta)|   & \left(2\xi_{11}^2(1-q_{11})^2I_{h_1}^2+\xi_{21}^2(1-q_{21})^2(S_{v_1}^2+I_{v_1}^2)+2\xi_{12}^2(1-q_{12})^2I_{h_2}^2+\xi_{22}^2(1-q_{22}^2)(S_{v_2}^2+I_{v_2}^2)\right)^{1/2}\\ \\
      &\leq \left(2[\xi_{11}^2(1-q_{11})^2+\xi_{12}^2(1-q-{12})^2]\left(\frac{\Lambda_H}{\mu_H}  \right)^2+2[\xi_{21}^2(1-q_{21})^2+\xi_{22}^2(1-q_{22})^2]\left(\frac{\Lambda_V}{\mu_V}  \right)^2\right)^{1/2} \\ \\
      &\leq \left( \xi\left[ \left(\frac{\Lambda_H}{\mu_H}\right)^2+\left(\frac{\Lambda_V}{\mu_V}\right)^2 \right]\right)=|f_{\bm\theta}(t,X,\bm\theta)|.
    \end{array}
\end{equation*}}
\end{proof}

\begin{lemma}\label{cota7}
The integrand of the performance index  satisfies
\begin{equation*}\label{eqcota7}
  f_0(t,\mathbf{X},\bm\theta) \geq \frac{1}{2}\min{\{d_1,d_2,d_3,d_4 \}}(\theta_{11}^2+\theta_{21}^2+\theta_{12}^2+\theta_{22}^2).
\end{equation*}
\end{lemma}
\begin{proof}
\begin{equation}\label{H5}
       \begin{array}{ll}
          f_0(t,\mathbf{X},\bm\theta) & = c_1(I_{h_1}+I_{v_1})+c_2(I_{h_2}+I_{v_2})+\frac{1}{2}(\theta_{11}^2+\theta_{21}^2+\theta_{12}^2+\theta_{22}^2) \\ \\
                 &\geq \frac{1}{2}(\theta_{11}^2+\theta_{21}^2+\theta_{12}^2+\theta_{22}^2) \\ \\
                 &\geq \frac{1}{2}\min{\{d_1,d_2,d_3,d_4 \}}(\theta_{11}^2+\theta_{21}^2+\theta_{12}^2+\theta_{22}^2).
       \end{array}
\end{equation}
\end{proof}
\begin{remark}\label{notacota7}
Hypothesis (H5) is fullfied by taking $b=2$, $c_2=0$ and $c_1=1/2\min{\{d_1,d_2,d_3,d_4 \}}$ in the last expression of (\ref{H5}).
\end{remark}

\subsection{Deduction of an optimal solution}
In this section, the Pontryaguin Principle for bounded controls \cite{lenhart2007optimal} is  used to compute the optimal controls of the problem (\ref{modelocontrolcomp}). First, let us observe that the Hamiltonian associated to (\ref{modelocontrolcomp}), is given by

\begin{equation}\label{hamiltonianomodelo1denso}
    \begin{array}{ll}
      &H(t, \mathbf{X}(t), \bm\theta(t), \mathbf{Z}(t))=f_0(t, \mathbf{X}, \bm\theta)+\mathbf{Z}\cdot f(t,\mathbf{X},\bm\theta)= \\ \\
      &c_1I_{h_1}+c_1I_{v_1}+c_2I_{h_2}+c_2I_{v_2}+\frac{1}{2}\left[d_1\theta_{11}^2+d_2\theta_{21}^2+d_3\theta_{12}^2+d_4\theta_{22}^2  \right]  \\\\
   +&z_1\left[ \Lambda_{h_1}+\omega_1 R_{h_1}-\left[\lambda_{11}\beta_{h_1}\epsilon_1\frac{I_{v_1}}{N_{h_1}}+\lambda_{12}\beta_{h_2}\epsilon_2\frac{I_{v_2}}{N_{h_2}}\right]S_{h_1}-\mu_{h_1}S_{h_1} \right] \\ \\
   +&z_2\left[\left[\lambda_{11}\beta_{h_1}\epsilon_1\frac{I_{v_1}}{N_{h_1}}+\lambda_{12}\beta_{h_2}\epsilon_2\frac{I_{v_2}}{N_{h_2}}\right]S_{h_1}-\xi_{11}\theta_{11}(t)(1-q_{11})I_{h_1}-(\delta_1+\rho_1+\mu_{h_1})I_{h_1} \right] \\ \\
   +&z_3\left[\xi_{11}\theta_{11}(t)(1-q_{11})I_{h_1}+\delta_1I_{h_1}-(\omega_1+\mu_{h_1})R_{h_1}  \right]\\\\
   +&z_4\left[\Lambda_{v_1}-\left[\lambda_{11}\beta_{v_1}\epsilon_1\frac{I_{h_1}}{N_{h_1}}+\lambda_{21}\beta_{v_2}\epsilon_2\frac{I_{h_2}}{N_{h_2}}  \right]S_{v_1} -\xi_{21}\theta_{21}(t)(1-q_{21})S_{v_1}-\mu_{v_1}S_{v_1} \right] \\ \\
   +&z_5\left[\left[\lambda_{11}\beta_{v_1}\epsilon_1\frac{I_{h_1}}{N_{h_1}}+\lambda_{21}\beta_{v_2}\epsilon_2\frac{I_{h_2}}{N_{h_2}}  \right]S_{v_1} -\xi_{21}\theta_{21}(t)(1-q_{21})I_{v_1}-\mu_{v_1}I_{v_1} \right]\\ \\
  + & z_6\left[\Lambda_{h_2}+\omega_2 R_{h_2}-\left[\lambda_{22}\beta_{h_2}\epsilon_2\frac{I_{v_2}}{N_{h_2}}+\lambda_{21}\beta_{h_1}\epsilon_1\frac{I_{v_1}}{N_{h_1}}\right]S_{h_2}-\mu_{h_2}S_{h_2} \right] \\ \\
  + & z_7
   \left[\left[\lambda_{22}\beta_{h_2}\epsilon_2\frac{I_{v_2}}{N_{h_2}}+\lambda_{21}\beta_{h_1}\epsilon_1\frac{I_{v_1}}{N_{h_1}}\right]S_{h_2}-\xi_{12}\theta_{12}(t)(1-q_{12})I_{h_2}-(\delta_2+\rho_2+\mu_{h_2})I_{h_2} \right]
   \\ \\
  + &z_8\left[\xi_{12}\theta_{12}(t)(1-q_{12})I_{h_2}+\delta_2I_{h_2}-(\omega_2+\mu_{h_2})R_{h_2} \right] \\ \\
  + &z_9\left[\Lambda_{v_2}-\left[\lambda_{22}\beta_{v_2}\epsilon_2\frac{I_{h_2}}{N_{h_2}}+\lambda_{12}\beta_{v_1}\epsilon_1\frac{I_{h_1}}{N_{h_1}}  \right]S_{v_2}-\xi_{22}\theta_{22}(t)(1-q_{22})S_{v_2} -\mu_{v_2}S_{v_2}  \right] \\ \\
  + &z_{10}\left[\left[\lambda_{22}\beta_{v_2}\epsilon_2\frac{I_{h_2}}{N_{h_2}}+\lambda_{12}\beta_{v_1}\epsilon_1\frac{I_{h_1}}{N_{h_1}}  \right]S_{v_2}-\xi_{22}\theta_{22}(t)(1-q_{22})I_{v_2}-\mu_{v_2}I_{v_2}  \right],
    \end{array}
\end{equation}
\noindent
where $\mathbf{Z}=(z_1, z_2,...,z_{10})$ is the vector of  \textit{adjoint variables} which determine the adjoint system. The adjoint system and the state equations of (\ref{modelocontrolcomp}) define  the optimal system. The main result of this section is summarized in the following theorem.

\begin{theorem}\label{teo6}
There are an optimal solution $\mathbf{X}^*(t)$ that minimize $J$ in $[0, T]$, and an adjoint vector of adjoint functions  $\mathbf{Z}$ such that
{\footnotesize
\begin{equation}\label{sistemaadjuntomodelo11}
\left\{ \begin{array}{ll}
&\dot{z_1} = \mu_{h_1}z_1+\lambda_{12}\beta_{h_2}\epsilon_2\frac{I_{v_2}}{N_{h_2}}(z_1-z_2)+\lambda_{11}\beta_{h_1}\epsilon_1I_{v_1}\frac{N_{h_1}-S_{h_1}}{N_{h_1}^2}(z_1-z_2) \\ &+\lambda_{11}\beta_{v_1}\epsilon_1\frac{I_{h_1}S_{v_1}}{N_{h_1}^2}(z_5-z_4)
+\lambda_{21}\beta_{h_1}\epsilon_1\frac{I_{v_1}S_{h_2}}{N_{h_1}^2}(z_7-z_6)+ \lambda_{12}\beta_{v_1}\epsilon_1 \frac{I_{h_1}S_{v_2}}{N_{h_1}^2}(z_{10}-z_9)\\ \\
 &\dot{z_2} = -c_1-[\xi_{11}\theta_{11}(1-q_{11})+\delta_1]z_3+\left[\delta_1+\rho_1+\mu_{h_1}-\xi_{11}\theta_{11}(1-q_{11})\right]z_2 \\
 &+\lambda_{11}\beta_{h_1}\epsilon_1\frac{I_{v_1}S_{h_1}}{N_{h_1}^2}(z_2-z_1)
 \lambda_{11}\beta_{v_1}\epsilon_1S_{v_1}\frac{(N_{h_1}-I_{h_1})}{N_{h_1}^2}(z_4-z_5)+\lambda_{21}\beta_{h_1}\epsilon_1\frac{I_{v_1}S_{h_2}}{N_{h_1}^2}(z_7-z_6)\\
 &+\lambda_{12}\beta_{v_1}\epsilon_1S_{v_2}\frac{N_{h_1}-I_{h_1}}{N_{h_1}^2}(z_9-z_{10}) \\ \\
& \dot{z_3}= -\omega_1z_1+(\omega_1+\mu_{h_1})z_3+\lambda_{11}\beta_{h_1}\epsilon_1\frac{I_{v_1}S_{h_1}}{N_{h_1^2}}(z_2-z_1)\\
    &+\lambda_{11}\beta_{v_1}\epsilon_1\frac{I_{h_1}S_{v_1}}{N_{h_1}^2}(z_5-z_4)+\lambda_{21}\beta_{h_1}\epsilon_1\frac{I_{v_1}S_{h_2}}{N_{h_1}^2}(z_7-z_6)+
\lambda_{12}\beta_{v_1}\epsilon_1\frac{I_{h_1}S_{v_2}}{N_{h_1}^2}(z_{10}-z_9)\\ \\
&\dot{z_4}= [\xi_{21}\theta_{21}(1-q_{21})+\mu_{v_1}]z_4+\left[\lambda_{11}\beta_{v_1}\epsilon_1\frac{I_{h_1}}{N_{h_1}}+\lambda_{21}\beta_{v_2}\epsilon_2\frac{I_{h_2}}{N_{h_2}}\right](z_4-z_5)\\
  \\ \\
&\dot{z_5}= -c_1+[\xi_{21}\theta_{21}(1-q_{21})+\mu_{v_1}]z_5+\lambda_{11}\beta_{h_1}\epsilon_1\frac{S_{h_1}}{N_{h_1}}(z_1-z_2)+\lambda_{21}\beta_{h_1}\epsilon_1\frac{S_{h_2}}{N_{h_1}}(z_6-z_7)\\ \\ &\dot{z_6}= -\mu_{h_2}z_6+\lambda_{12}\beta_{h_2}\epsilon_2\frac{I_{v_2}S_{h_1}}{N_{h_2}^2}(z_2-z_1)+\lambda_{21}\beta_{v_2}\epsilon_2\frac{I_{h_2}S_{v_1}}{N_{h_2}^2}(z_5-z_4) \\ &+\lambda_{22}\beta_{h_2}\epsilon_2I_{v_2}\frac{N_{h_2}-S_{h_2}}{N_{h_2}^2}(z_6-z_7)+\lambda_{21}\beta_{h_1}\epsilon_1\frac{I_{v_1}}{N_{h_1}}(z_6-z_7)+
  \lambda_{22}\beta_{v_2}\epsilon_2\frac{I_{h_2}S_{v_2}}{N_{h_2}^2}(z_9-z_{10})\\
 \\  &\dot{z_7}=-c_1-(\xi_{12}\theta_{12}(1-q_{12})+\delta_2)z_8+[\xi_{12}\theta_{12}(1-q_{12})+\delta_2+\rho_2+\mu_{h_2}]z_7 \\&+\lambda_{12}\beta_{h_2}\epsilon_2\frac{I_{v_2}S_{h_1}}{N_{h_2}^2}(z_2-z_1)+
 \lambda_{21}\beta_{v_2}\epsilon_2S_{v_1}\frac{(N_{h_2}-I_{h_2})}{N_{h_2}^2}(z_5-z_4) \\&+\lambda_{22}\beta_{h_2}\epsilon_2\frac{I_{v_2}S_{h_2}}{N_{h_2}^2}(z_7-z_6)+
 \lambda_{22}\beta_{v_2}\epsilon_2S_{v_2}\frac{N_{h_2}-I_{h_2}}{N_{h_2}^2}(z_9-z_{10}) \\
   \\ &\dot{z_8}= -\omega_2z_6+(\omega_2+\mu_{h_2})z_8+\lambda_{12}\beta_{h_2}\epsilon_2\frac{I_{v_2}S_{h_1}}{N_{h_2}^2}(z_2-z_1)+
    \\&\lambda_{21}\beta_{v_2}\epsilon_2\frac{I_{h_2}S_{v_1}}{N_{h_2}^2}(z_5-z_4)+\lambda_{22}\beta_{h_2}\epsilon_2\frac{I_{v_2}S_{h_2}}{N_{h_2}^2}(z_7-z_6)+
\lambda_{22}\beta_{v_2}\epsilon_2\frac{I_{h_2}S_{v_2}}{N_{h_2}^2}(z_{10}-z_9)\\
\\ &\dot{z_9}=[\xi_{22}\theta_{22}(1-q_{22})+\mu_{v_2}]z_9+\left[\lambda_{22}\beta_{v_2}\epsilon_2\frac{I_{h_2}}{N_{h_2}}+\lambda_{12}\beta_{v_1}\epsilon_1\frac{I_{h_1}}{N_{h_1}}\right](z_9-z_{10})\\
  \\
\\ &\dot{z_{10}}= -c_2+[\xi_{22}\theta_{22}(1-q_{22})+\mu_{v_2}]z_{10}+\lambda_{22}\beta_{h_2}\epsilon_2\frac{S_{h_2}}{N_{h_2}}(z_6-z_7)+\lambda_{12}\beta_{h_2}\epsilon_2\frac{S_{h_1}}{N_{h_2}}(z_1-z_2),\\
\\
 \end{array}
  \right.
\end{equation}}

with transversality condition $\mathbf{Z}(t)=\mathbf{0}$ and the following characterization of the controls
\begin{equation}\label{micaracterizacionmodelo1}
 \left\{ \begin{array}{lll}
                \theta_{11}^*= & \min\left\lbrace\max\left\lbrace0,\frac{\xi_{11}(1-q_{11})I_{h_1}(z_2-z_3)}{d_1} \right\rbrace,1\right \rbrace    \\
             \\ \theta_{21}^*= & \min\left\lbrace\max\left\lbrace0, \frac{\xi_{21}(1-q_{21})(S_{v_1}z_4+I_{v_1}z_5)}{d_2}\right\rbrace, 1\right \rbrace    \\
             \\ \theta_{12}^*= & \min\left\lbrace\max\left\lbrace0, \frac{\xi_{12}(1-q_{12})I_{h_2}(z_7-z_8)}{d_3}  \right\rbrace, 1\right\rbrace  \\
             \\ \theta_{22}^*= &  \min\left\lbrace\max\left\lbrace0, \frac{\xi_{22}(1-q_{22})(S_{v_2}z_9+I_{v_2}z_{10})}{d_4}  \right\rbrace, 1\right\rbrace  .
             \end{array}
   \right.
\end{equation}
\end{theorem}

\begin{proof}
The  Pontryaguin Principle guarantees the existence of the vector of adjoint variables $\mathbf{Z}$
whose components satisfy
\begin{eqnarray}\label{echamiltonmodelo1denso}
   \dot{z_i}=\frac{dz_i}{dt}&=&-\frac{\partial H}{\partial x_i} \nonumber\\
   z_i(T) &=& 0, \quad i=1,2,...,10 \nonumber \\
   H(t, \mathbf{X}, \bm\theta, \mathbf{Z}) &=& \max_{\bm\theta\in U}H(t, \mathbf{X}, \bm\theta, \mathbf{Z}).
\end{eqnarray}
\noindent
Thus, the derivatives of the adjoint variables are
\begin{align*}
  \dot{z_1} &= -\frac{\partial H}{\partial S_{h_1}}, \quad z_1(T)=0 & \dot{z_6} &= -\frac{\partial H}{\partial S_{h_2}}, \quad z_6(T)=0  \\
  \dot{z_2} &= -\frac{\partial H}{\partial I_{h_1}}, \quad z_2(T)=0   & \dot{z_7} &= -\frac{\partial H}{\partial I_{h_2}}, \quad z_7(T)=0 \\
  \dot{z_3} &= -\frac{\partial H}{\partial R_{h_1}}, \quad z_3(T)=0   & \dot{z_8} &= -\frac{\partial H}{\partial R_{h_2}}, \quad
   z_8(T)=0 \\
   \dot{z_4} &= -\frac{\partial H}{\partial S_{v_1}}, \quad z_4(T)=0   & \dot{z_9} &= -\frac{\partial H}{\partial S_{v_2}}, \quad
  z_9(T)=0 .\\
  \dot{z_5} &= -\frac{\partial H}{\partial I_{v_1}}, \quad z_5(T)=0   & \dot{z_{10}} &= -\frac{\partial H}{\partial I_{v_2}}, \quad
  z_{10}(T)=0 .\\
\end{align*}
\noindent
Replacing the derivatives of $H$ with respect to the state equations in above equalities, we obtain  the system (\ref{sistemaadjuntomodelo11}). Additionally, the optimality conditions for the Hamiltonian are given by
\begin{equation*}\label{derivadahamiltoniamoodelo1}
    \frac{\partial H}{\partial \theta_{11}^*}= \frac{\partial H}{\partial \theta_{12}^*}= \frac{\partial H}{\partial \theta_{21}^*}= \frac{\partial H}{\partial \theta_{22}^*}=0,
\end{equation*}
\noindent
from where

 \begin{eqnarray*}\label{usubistar}
   \theta_{11}^* &=&  \frac{\xi_{11}(1-q_{11})I_{h_1}(z_2-z_3)}{d_1} \nonumber\\
   \theta_{21}^* &=& \frac{\xi_{21}(1-q_{21})(S_{v_1}z_4+I_{v_1}z_5)}{d_2}\nonumber\\
   \theta_{12}^* &=& \frac{\xi_{12}(1-q_{12})I_{h_2}(z_7-z_8)}{d_3} \nonumber \\
   \theta_{22}^* &=& \frac{\xi_{22}(1-q_{22})(S_{v_2}z_9+I_{v_2}z_{10})}{d_4}.
 \end{eqnarray*}
In consequence,  $\theta_{11}^*$ satisfies

\begin{equation*}\label{u1star}
\theta_{11}^*= \left\{ \begin{array}{ccc}
             1 &   if  &\frac{\xi_{11}(1-q_{11})I_{h_1}(z_2-z_3)}{d_1} >0 \\ \frac{\xi_{11}(1-q_{11})I_{h_1}(z_2-z_3)}{d_1}  &  if & 0\leq \frac{\xi_{11}(1-q_{11})I_{h_1}(z_2-z_3)}{d_1}\leq1\\
             0 &  if  & \frac{\xi_{11}(1-q_{11})I_{h_1}(z_2-z_3)}{d_1} <0,
             \end{array}
   \right.
\end{equation*}
 or equivalently
\begin{equation}\label{u1starcompacto}
    \theta_{11}^*=  \min\left\lbrace\max\left\lbrace0,\frac{\xi_{11}(1-q_{11})I_{h_1}(z_2-z_3)}{d_1} \right\rbrace,1\right \rbrace.
\end{equation}
Using a similar reasoning for $\theta_{21}^*$, $\theta_{12}^*$ and $\theta_{22}^*$ we obtain the characterization (\ref{micaracterizacionmodelo1}) which completes the proof.
\end{proof}

\section{Numerical experiments}

In this section  we present some numerical simulations associated with the implementation of drugs and insecticides as control strategies,  as well as their effects on the infected individuals under uncoupled and strongly--coupling scenarios.
For the simulations, we use the \textit{forward-backward sweep method} proposed  by Lenhart and Workman \cite{lenhart2007optimal}. The implementation time of the control strategies will be approximately 10 days, which is the duration of a malaria treatment. The values of the relative weights associated with the control, will be those of  Table 7 from \cite{romero2018optimal}.
\par
Figure \ref{Control_desacoplado} shows the behavior of the infected individuals in patches 1 and 2 under uncoupled scenario. Due to in this scenario, the disease only remains in patch 1 and does not spread to patch 2,  the density of infected individuals decreases with control in patch 1 and the effects of the controls in patch 2 are not necessary.
\begin{figure}[h]
\centering
\includegraphics[width=7 cm, height=7.0 cm]{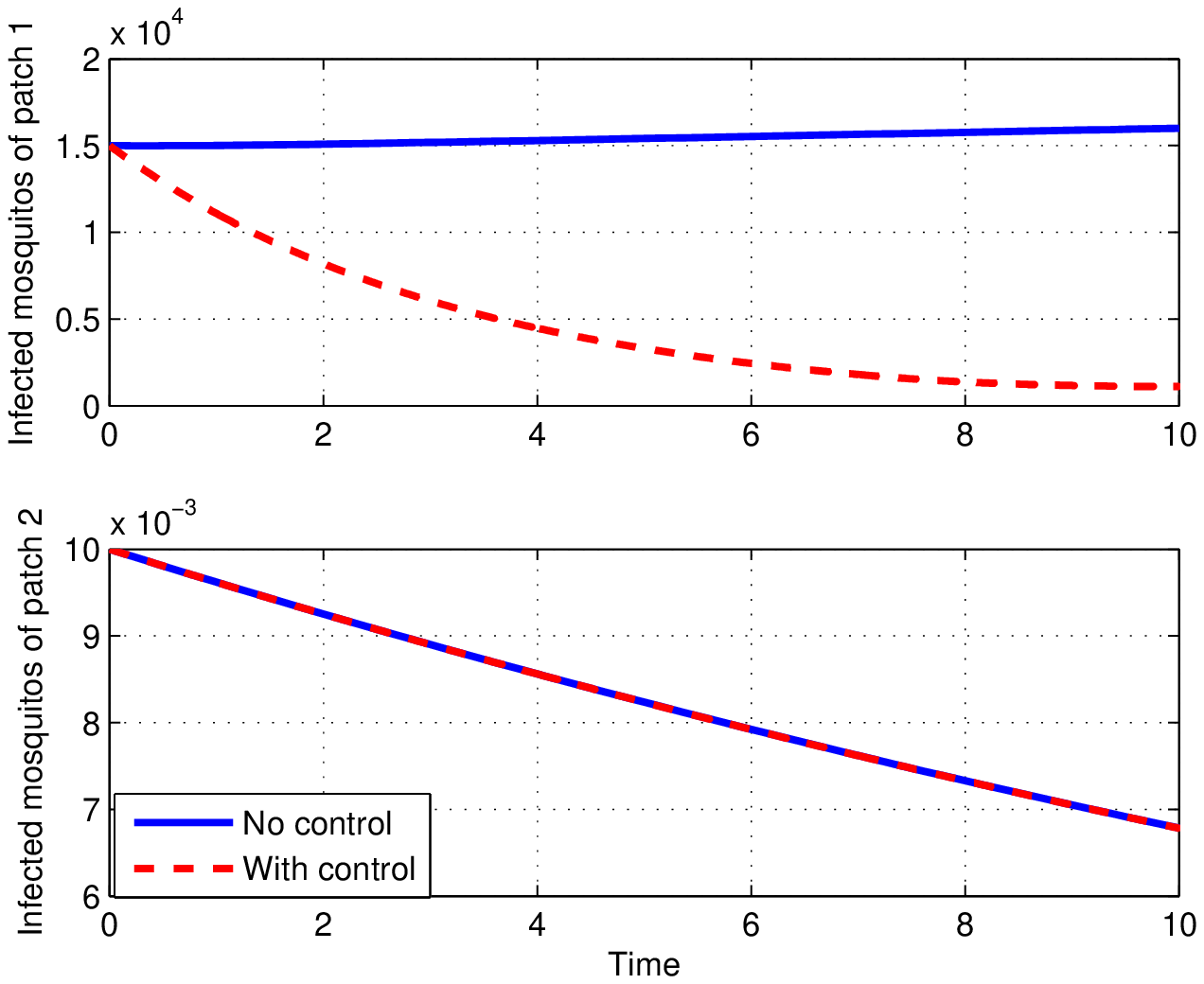}
\includegraphics[width=7 cm, height=7.0 cm]{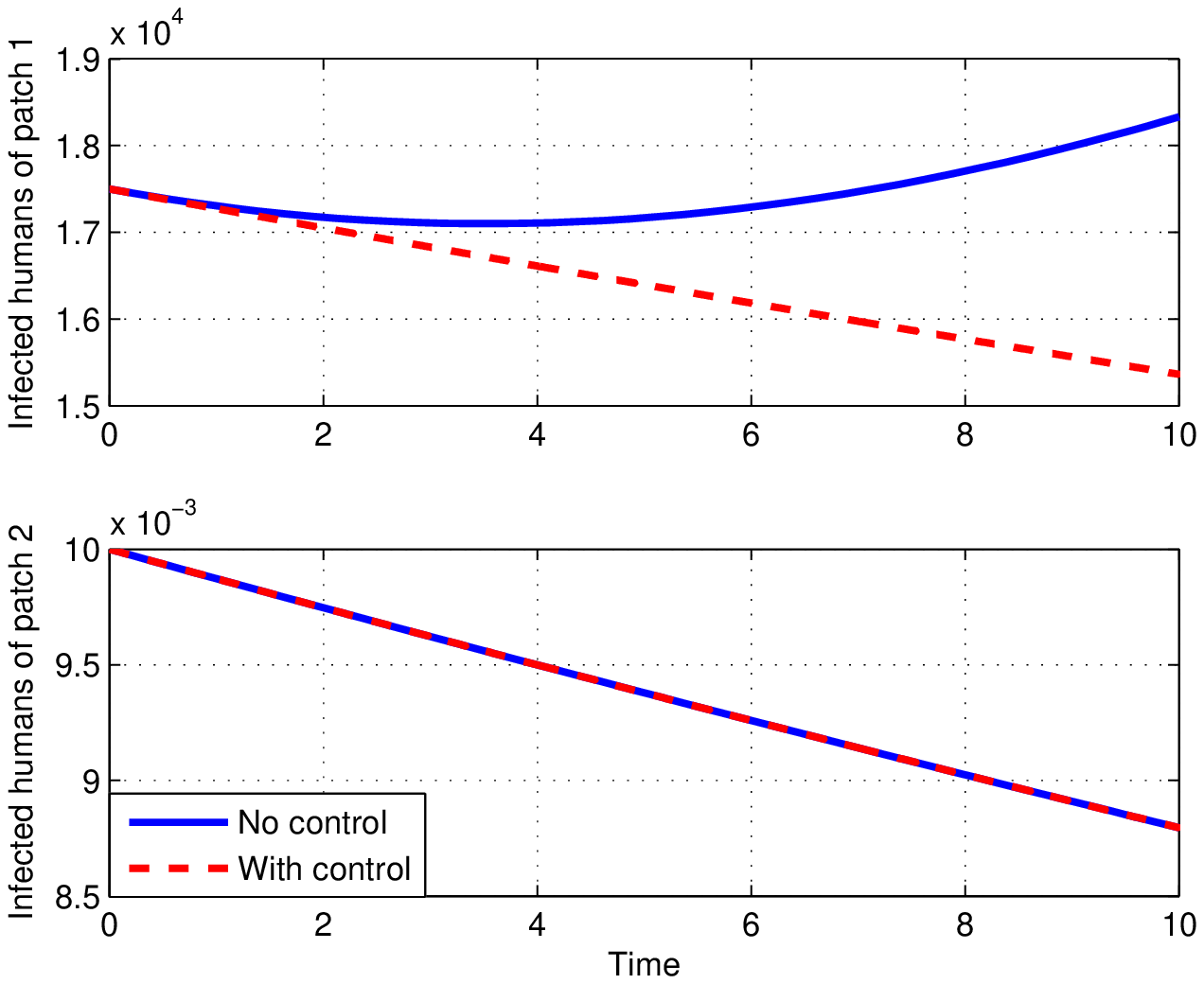} \\
\includegraphics[width=7 cm, height=3.5 cm]{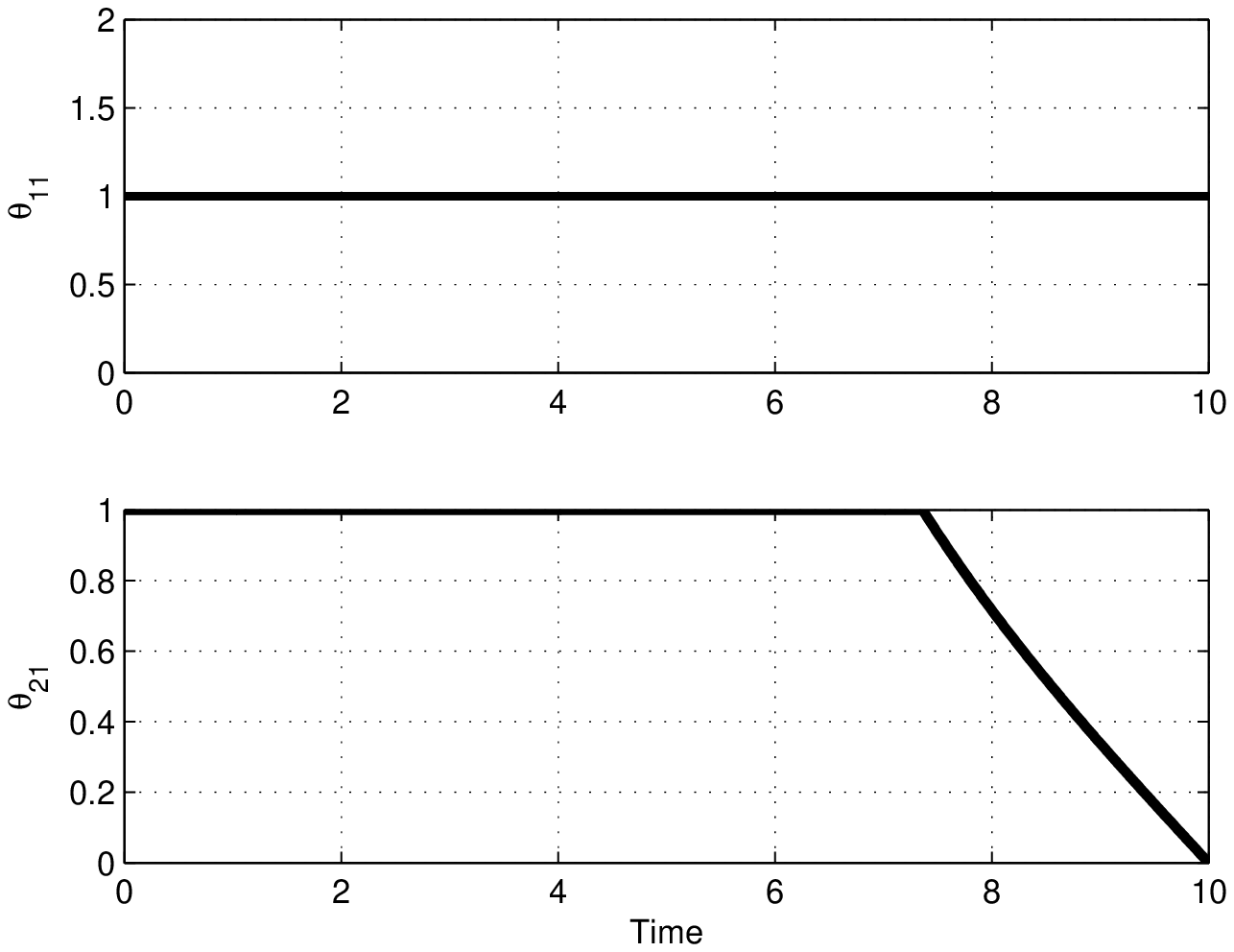}
\includegraphics[width=7 cm, height=3.5 cm]{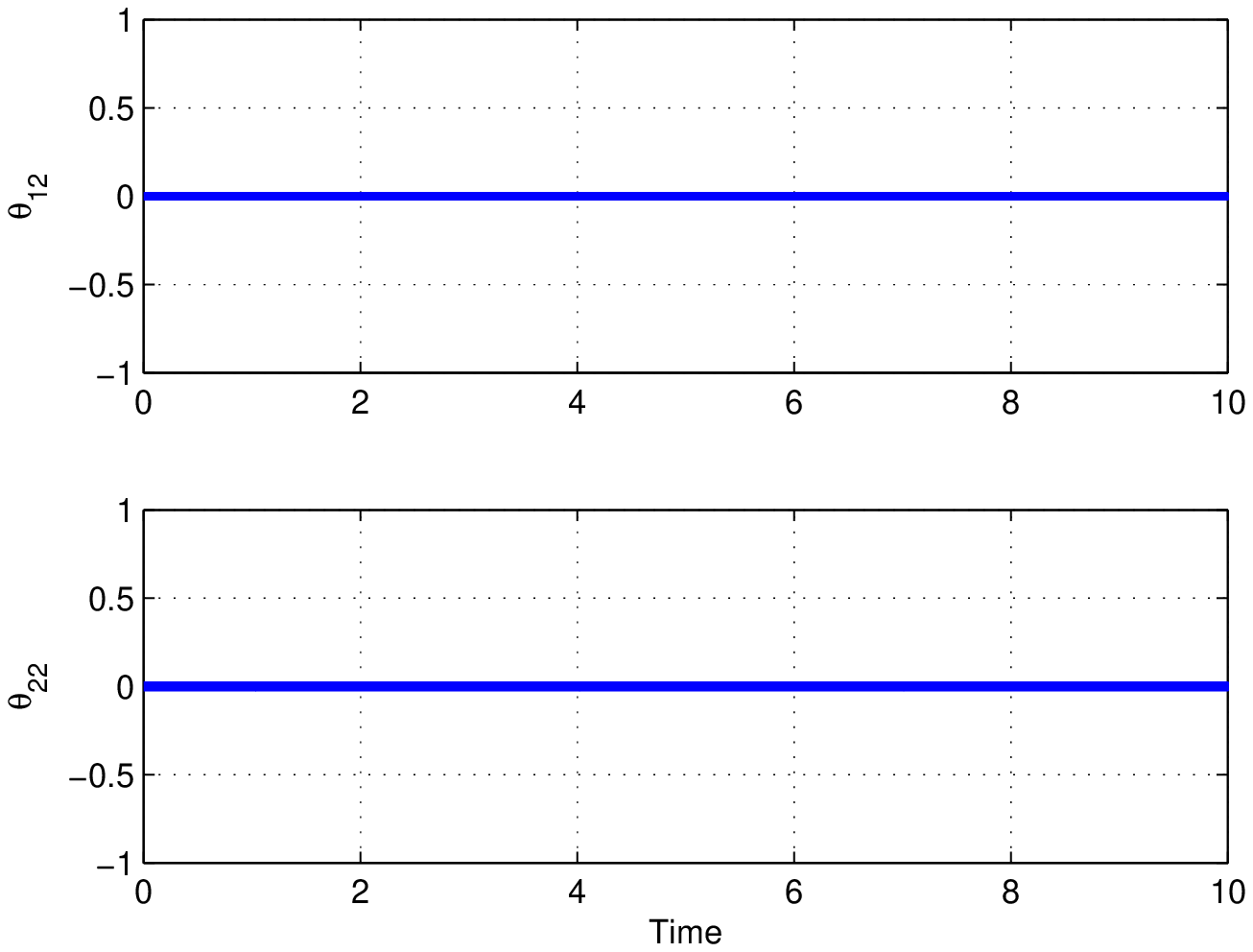}
\caption{{\scriptsize Control under uncoupled scenario.}}
\label{Control_desacoplado}
\end{figure}
In Figure \ref{Control_strongly} we can see the behavior of infected individuals in patches 1 and 2 under strongly--coupling scenario. Here, the infection decreases with control in both patches,
but the efforts are greater in patch 1 than in patch 2.
\begin{figure}[h]
\centering
\includegraphics[width=7 cm, height=7.0 cm]{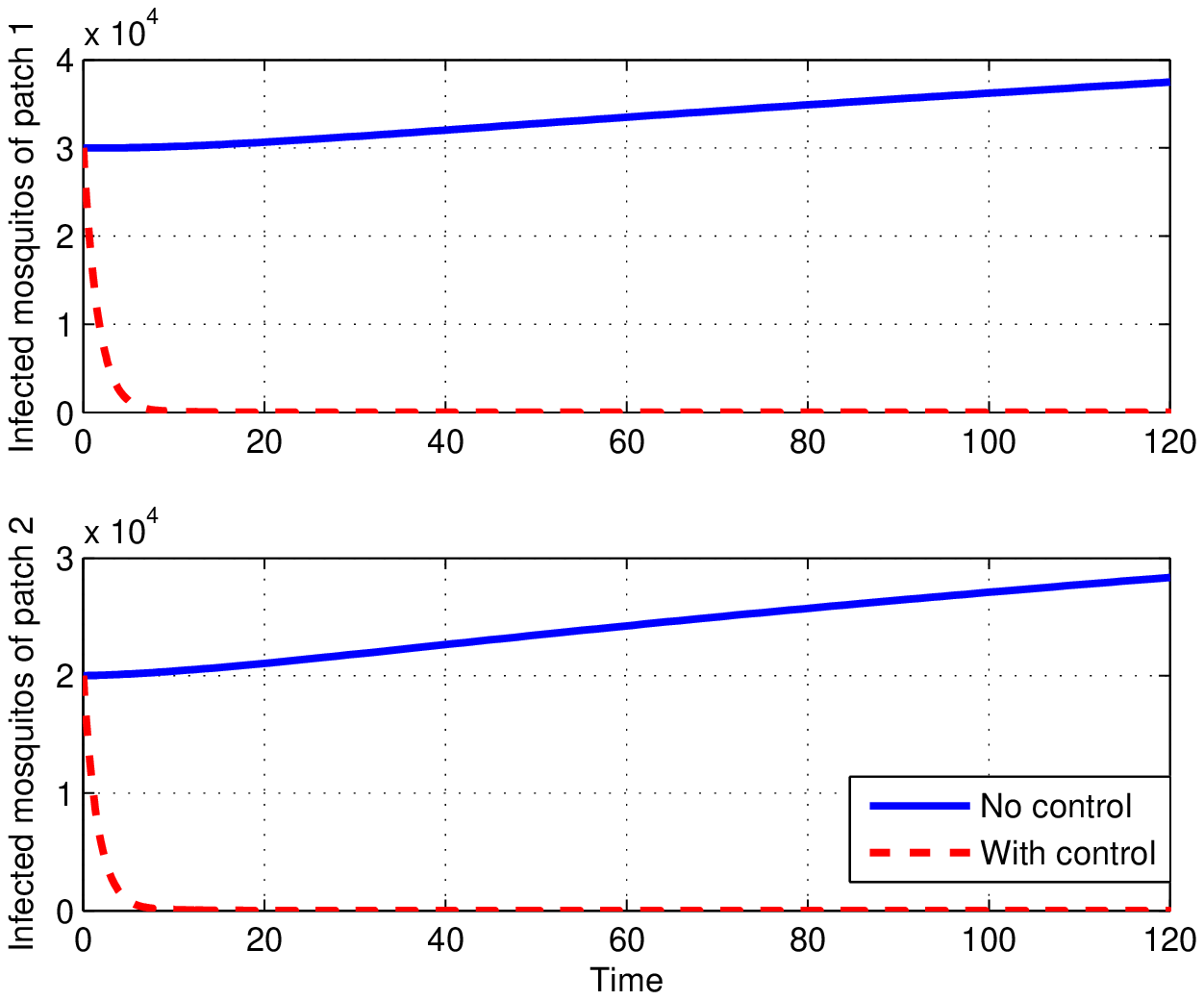}
\includegraphics[width=7 cm, height=7.0 cm]{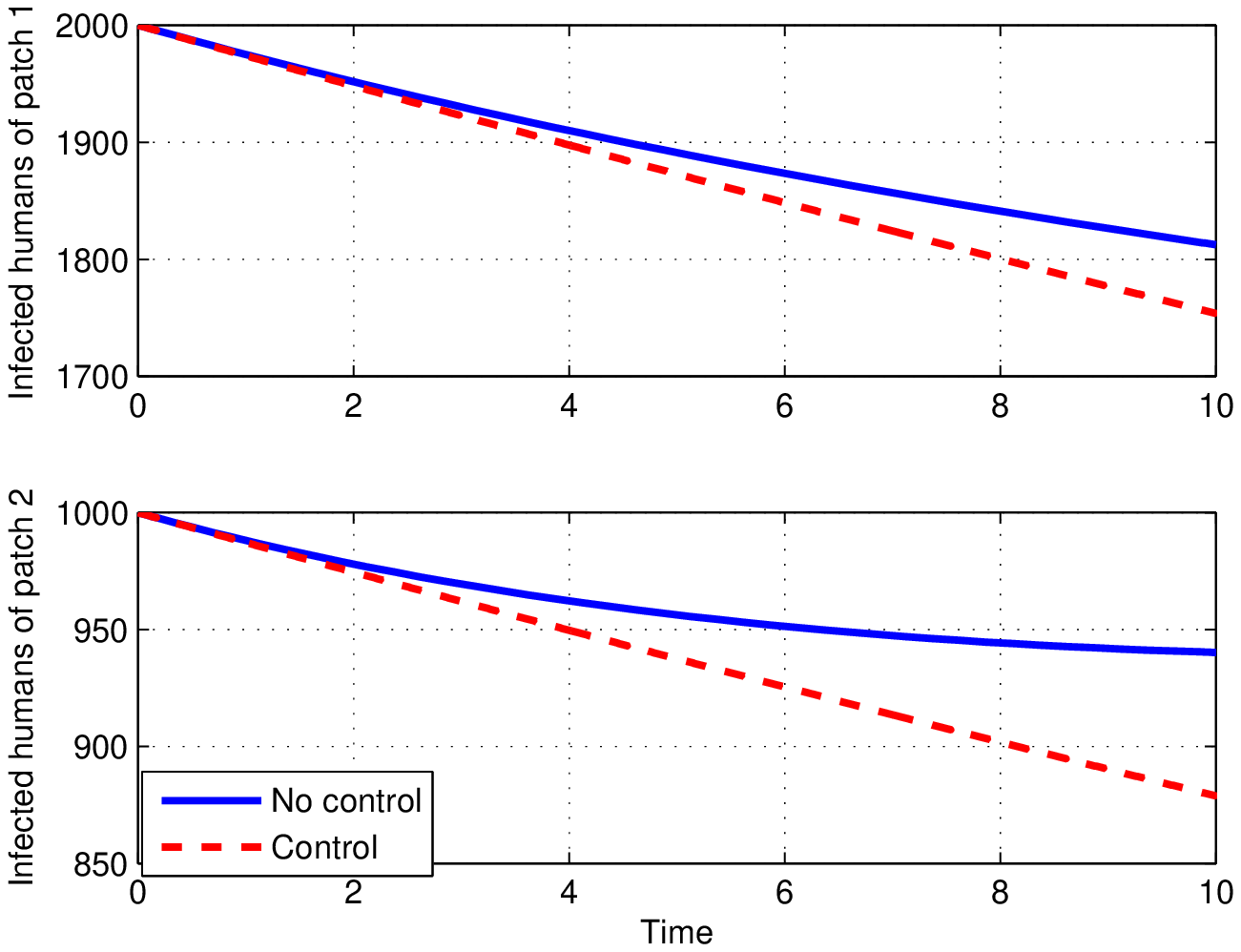} \\
\includegraphics[width=7 cm, height=3.5 cm]{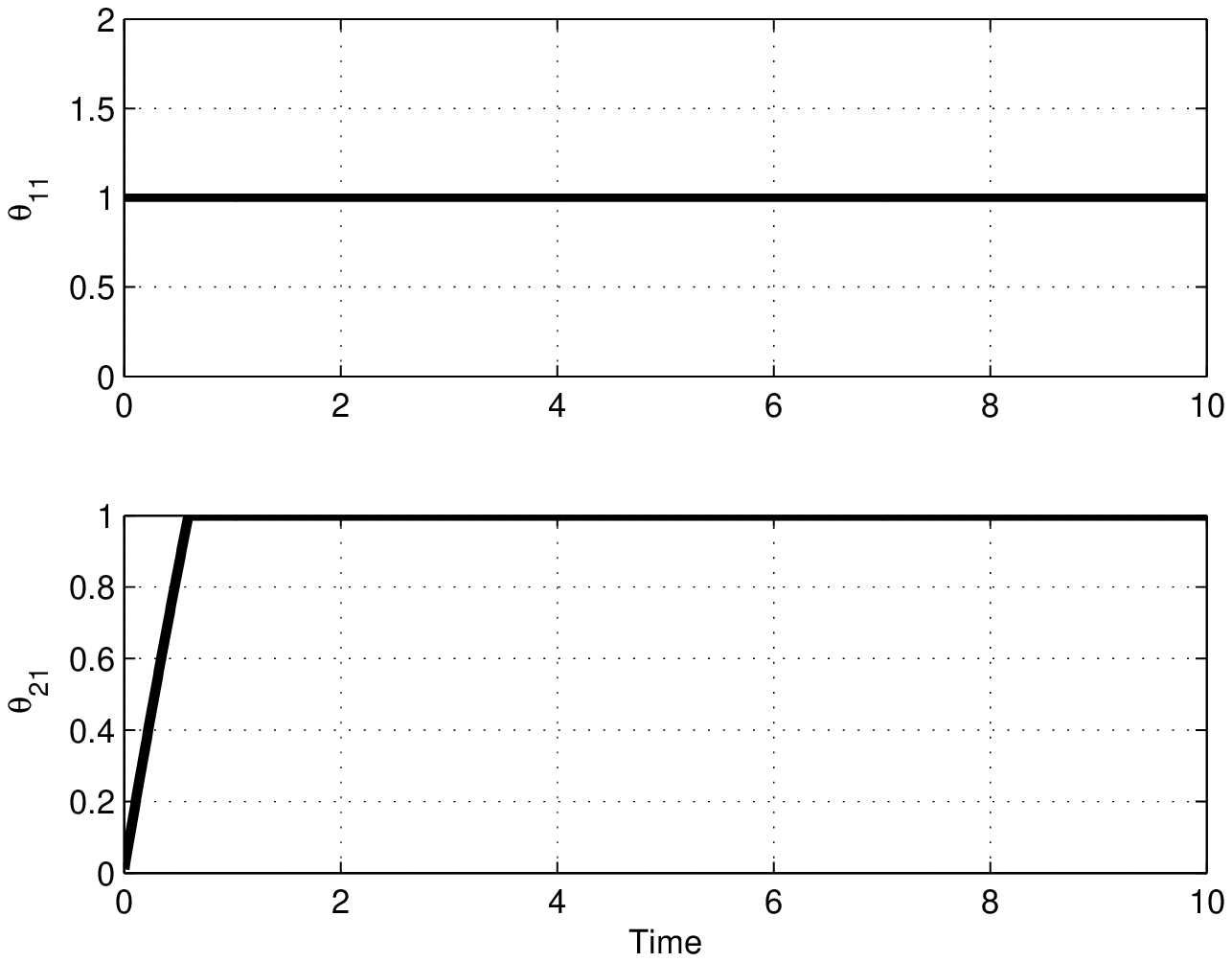}
\includegraphics[width=7 cm, height=3.5 cm]{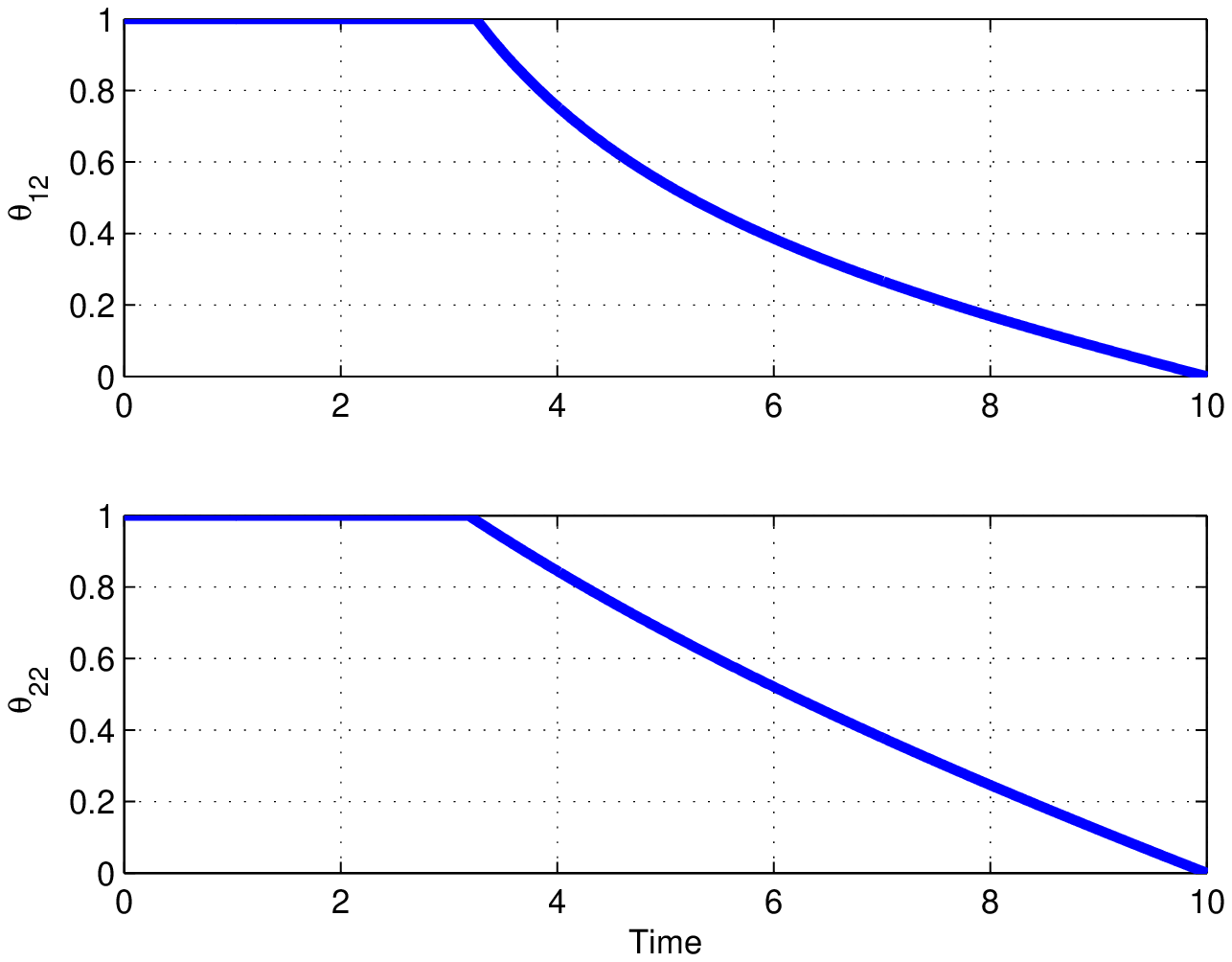}
\caption{{\scriptsize Control under strongly--coupled scenario.}}
\label{Control_strongly}
\end{figure}
\par
In both cases, uncoupled and strongly--coupling scenario,  the effects of the controls
are highly effective and fast to eliminate the disease in patch 1, while in patch 2 the elimination depends of the coupling scenario.

\section{Discusion}

In this work,  we model the malaria transmission dynamics, considering three factors that hinder its control: resistance to drugs, resistance to  insecticides and population movement.
\noindent
To illustrate the above factors, we divide our work into two mathematical models. (a) A mathematical model in a patch under the hypothesis that the  parasites are resistant to the drugs, and the  mosquitoes are resistant to the insecticides. In this first model, we make a qualitative analysis of the solutions of the system, which reveal the existence of a forward bifurcation and the global stability of the DFE. From the biological point of view, the existence of a forward bifurcation indicates that the disease can be controlled by keeping the local $\mathcal{R}_{0_{one}} $ below of one. Since the expression for $ \mathcal{R}_{0_{one}} $ given on (\ref{R0one}) depends directly on the resistance acquisition ratios $ q_1 $ and $ q_2 $, then at lower levels of resistance acquisition, the value of $ \mathcal{R}_{0_{one}} $ decreases, which implies that the infection levels decrease. On the other hand, since $ \mathcal{R}_{0_{one}} $ depends inversely on the effects of the drugs and insecticides, then an increase in the recovery rate of humans due to drugs and the death of mosquitoes due insecticides, implies a decrease of $ \mathcal{R}_{0_{one}} $ and therefore the burden of infection.
\noindent
The numerical experiments for this first model corroborate the theoretical results. Here,  we assume that the infected patients are treated with ACT (artemisinin--based combination therapy) and to contrast the fumigation of mosquitoes with deltamethrin and cifluthrin, where the first insecticide is more effective than the second one. With total resistance to the drugs and insecticides ($ q_1 = q_2 = 1 $), we verify that the burden of infection persists regardless of the type of drug and insecticide used, while without resistance ($ q_1 = q_2 = 0 $), the burden of infection decreases with the use of deltamethrin and is maintained at low levels with the use of cifluthrin. These results are alarms in  public health, because despite the pharmaceutical industry is taking care day after day to create new drugs and new insecticides, if the phenomenon of resistance acquisition is not counteracted, the problem of malaria control will be increasingly difficult, and in some cases impossible.
\par
(b) For the model in two patches, we consider the same hypotheses of the model in a single patch, and additionally, movement of populations between  two patches. For this case, we determine the global basic reproductive number $ \mathcal{R}_0 $, and through numerical experiments, we illustrate the behavior of the solutions when the infection starts in the patch 1 (rural areas of Tumaco \cite{romero2018optimal}), and under three coupling scenarios: (1) uncoupled scenario.  When there is no movement  between patches, the infection remains endemic in the patch 1 and does not spread to the patch 2. (2) Weakly--coupling. If the  probabilities of visiting  between both patches are low, the disease is endemic in the patch 1 and remains at a very low load in the patch 2. (3) Strongly--coupling. If the probabilities of visiting  between both patches is high, the disease remains endemic in both patches.   These results corroborate the phenomenon of reinfection in areas where malaria has been eradicated and is not endemic, as is the case of urban malaria. Here, a new alarm in public health is created, because if malaria has been completely eradicated in a sector and is not endemic there, the movement of humans (or mosquitoes) from endemic areas can activate the infection alarm again.
\par
Finally, using results of a local sensitivity analysis of parameters to the global $ \mathcal{R}_0 $, we formulated an optimal control problem by using of  drugs and insecticides as control strategies. The results of the theoretical and numerical analysis of the optimal control problem reveal that under uncoupled scenario, the control is effective and necessary in patch 1 but not in patch 2, while under strongly--coupling, greater efforts are required to control the disease in patch 1 than in patch 2.
\par
An open problem through this research is to incorporate prophylaxis as a control strategy for the disease, that is, patient education campaigns both in the use of drugs and in the use of insecticides. In this way, the resistance phenomenon will be mitigated and the control campaigns for the disease will be more effective and less expensive.


\end{document}